\documentclass[letterpaper,twocolumn,10pt]{article}
\usepackage{usenix}

\usepackage{tikz}
\usepackage{amsmath}
\usepackage{subcaption}
\usepackage[export]{adjustbox}
\usepackage{booktabs}
\usepackage{tikz}
\usepackage{amsmath}
\usepackage{cleveref}
\usepackage{url}
\usepackage{caption}
\usepackage[most]{tcolorbox}
\usepackage{enumitem}
\usepackage{amsthm}
\usepackage{amsfonts}
\newtheorem{definition}{Definition}

\usepackage{graphicx}
\usetikzlibrary{decorations.pathreplacing, arrows, fit}
\usetikzlibrary{arrows.meta}
\usetikzlibrary{automata, positioning, arrows,calc,hobby,shapes,shapes.geometric,arrows,matrix}
\usetikzlibrary{shapes.multipart}
\usepackage{tkz-euclide}
\tikzstyle{int}=[draw, fill=lightgray, minimum height=1cm, minimum width=1.4cm,text width=1.3cm,align =center]
\definecolor{lightgreen}{RGB} {172, 243, 174}
\definecolor{lightred}{RGB}{250, 107, 132}
\definecolor{lightyellow}{RGB}{255, 255, 167}
\definecolor{lightgrey}{RGB}{220, 220, 220}
\tikzset{cross/.style={cross out, draw=black, ultra thick,fill=none, minimum size=4*(#1-\pgflinewidth), inner sep=0pt, outer sep=0pt}, cross/.default={3.5pt}}
\usepackage{thm-restate}
\newcommand\Paragraph[1]{\smallskip \noindent\textbf{#1}.}

\AtBeginDocument{%
  \providecommand\BibTeX{{%
    \normalfont B\kern-0.5em{\scshape i\kern-0.25em b}\kern-0.8em\TeX}}}
\begin{document}
%-------------------------------------------------------------------------------

%don't want date printed
\date{}

\title{The Potential of Self-Regulation for Front-Running Prevention on Decentralized Exchanges}
%for single author (just remove % characters)
\author{
{\rm Lioba Heimbach}\\
ETH Zurich\\
hlioba@ethz.ch
\and
{\rm Eric Schertenleib}\\
unaffiliated\\
eric.schertenleib@gmail.com
% copy the following lines to add more authors
 \and
 {\rm Roger Wattenhofer}\\
 ETH Zurich\\
wattenhofer@ethz.ch 
} % end author

\maketitle

%-------------------------------------------------------------------------------
\begin{abstract}
%-------------------------------------------------------------------------------
The transaction ordering dependency of the smart contracts building \emph{decentralized exchanges (DEXes)} allow for predatory trading strategies. In particular, front-running attacks present a constant risk for traders on DEXes. Whereas legal regulation outlaws most front-running practices in traditional finance, such measures are ineffective in preventing front-running on DEXes. While novel market designs hindering front-running may emerge, it remains unclear whether the market's participants, in particular, liquidity providers, would be willing to adopt these new designs. A misalignment of the participant's private incentives and the market's social incentives can hinder the market from adopting an effective prevention mechanism. 

We present a game-theoretic model to study the behavior of sophisticated traders, retail traders, and liquidity providers in DEXes. Sophisticated traders adjust for front-running attacks, while retail traders do not, likely due to lack of knowledge or irrationality. Our findings show that with less than $1\%$ of order flow from retail traders, traders' and liquidity providers' interests align with the market’s social incentives -- eliminating front-running attacks. However, the benefit from embracing this novel market is often small and may not suffice to entice them. With retail traders making up a larger proportion (around $10\%$) of the order flow, liquidity providers tend to stay in pools that do not protect against front-running. This suggests both educating traders and providing additional incentives for liquidity providers are necessary for market self-regulation.
\end{abstract}

%-------------------------------------------------------------------------------
\section{Introduction}
%-------------------------------------------------------------------------------

The emergence of \emph{decentralized finance (DeFi)} on Ethereum~\cite{wood2014ethereum} greatly enhanced the interest in cryptocurrency applications. DeFi is a blockchain-based form of finance that utilizes \emph{smart contracts} to offer many traditional financial instruments, but without relying on financial intermediaries. A prime example thereof is \emph{decentralized exchanges (DEXes)}. While traditional exchanges match individual buyers and sellers with the limit order book mechanism, a DEX algorithmically sets the exchange rate for a trade. To this end, DEXes store liquidity for exchanges between individual cryptocurrency pairs in smart contracts, referred to as \emph{liquidity pools}. The trade size and the respective cryptocurrency pair's amount and ratio of reserves control the price. The pool charges a fee for every trade which is proportional to the trade's input amount and distributed pro-rata amongst the pool's \emph{liquidity providers}. 

DEXes are becoming increasingly popular. Yet, the rise of DEXes does not come without caveats, leading to the characterization of the Ethereum peer-to-peer network as a dark forest. Predatory trading bots prey on user transactions in Ethereum's \emph{mempool}, the public waiting area for transactions. Predatory trading schemes exploit the lack of privacy given to transactions prior to their execution. Moreover, the smart contracts that build DEXes are dependent on the transaction order. Generally, these attacks involve the attacker \emph{front-running} a victim's transaction. One of the most frequently observed strategies exploiting this dependency on transaction ordering is the \emph{sandwich attack}~\cite{qin2021quantifying} which we describe in Section~\ref{sec: Preliminaries}. We focus on sandwich attacks in this work, as this is the only front-running attack directly impacting the welfare of liquidity providers.\footnote{Apart from sandwich attacks, there are destructive front-running attacks~\cite{eskandari2019sok}. Thereby, the attacker searches for trades that exploit arbitrage opportunities and front-runs these and is indifferent to whether the victim's transaction executes. The DEXes volume remains unchanged.} Such an attack occurs when a trading bot front- and back-runs a victim's transaction, forcing the trade to execute at an unfavorable price. Between 1 August and 31 August 2024, more than 130,000 transactions were sandwiched on Ethereum blockchain's DEXes~\cite{2022eigenphi}. 

Given the severity of front-running attacks on DEXes, the market is seeking mechanisms that can prevent such attacks. While front-running attacks are outlawed in traditional finance, the anonymity of market participants and the absence of a central authority do not allow for an effective regulatory approach for DEXes. Therefore, they require new innovative solutions to prevent front-running attacks. In response, multiple approaches to prevent front-running and, more generally, transaction reordering manipulation are currently under development~\cite{heimbach2022sok}. These approaches generally ensure that transaction contents are private from the public until an order is agreed upon. Further, some designs are already being adopted~\cite{2021flash,2021eden,2022bloXroute}.

Even though successful market designs preventing front-running attacks increase market efficiency as a whole~\cite{manahov2016front}, there may still be insufficient incentives to adopt these prevention schemes. Liquidity providers, who potentially benefit from the additional trading volume from sandwich attacks, might be reluctant to embrace such a market and, similarly, some traders might be slow to adapt to such new markets. Examples, where the divergence of private and social incentives has led to adoption failure exist in traditional limit order book exchanges~\cite{budisch2019theory}. 

In this work, we develop a game-theoretical model to study whether traders and liquidity providers would embrace DEXes that incorporate a front-running prevention mechanism. Our repeated game consists of two hypothetical liquidity pools, one allowing sandwich attacks and one implementing a sandwich attack prevention scheme. Traders and liquidity providers distribute across the two pools, thereby maximizing their private incentives. The objective of liquidity providers is to maximize their fees earned while the traders seek to execute their orders at the best possible price. Thus, while the presence of sandwich attacks leads to additional trades from the attack, it reduces the volume of ordinary traders as they receive a poor price.

To capture different behavior traders, we distinguish between two types of traders. A portion of the trade orders originates from \textit{sophisticated} traders who are aware of sandwich attacks and adjust their behavior accordingly (fully rational), whereas, another fraction of trades stem from \textit{retail} traders who are either oblivious or indifferent to these attacks. This distinction is made to account for information asymmetry and a degree of irrationality present in the behavior of users.

\Paragraph{Contributions} We summarize our contributions below: 

\begin{itemize}[leftmargin=*,topsep=0pt,itemsep=-1ex,partopsep=1ex,parsep=1ex]
    \item \textit{Small Retail Volume.} We find that a parameter regime exists for which the players utilize the liquidity pool with front-running prevention -- indicating that the market can fix itself. This regime dominates when we have a very small proportion of order flow stemming from retail traders. However, even when the private incentives of liquidity providers align with the market's social incentives, the benefit of embracing the new market can be small. 
    \item \textit{Significant Retail Volume.} When the proportion of retail traders increases, i.e., we observe more irrational behavior from traders, market conditions where the private incentives of liquidity providers and the system's social incentives are misaligned become more common.
    \item \textit{Moving Towards Self-Regulation.} Finally, we highlight the importance of educating traders and point towards additional incentives the market could provide to entice liquidity providers to adjust their strategy. 
\end{itemize}

%-------------------------------------------------------------------------------
\section{Related Work}
%-------------------------------------------------------------------------------
Front-running has long been prevalent in traditional finance~\cite{BERNHARDT2008front,danthine1998front,comerton2007anonymity} where the regulator is tasked with banning such practices~\cite{markham1988front,moosa2015regulation}. The lack of a central authority in DeFi, however, means that the market must regulate itself. Thus, we study whether the private incentives of market participants obstruct the adoption of innovative market designs preventing front-running. 

Eskandir et al.~\cite{eskandari2019sok} are the first to systematize work surrounding front-running on DeFi. In a similar line of research, Daian et al.~\cite{daian2020flash} study the risks of front-running on DEXes. They observe traditional forms of predatory trading behaviors adapting to the blockchain ecosystem. Park~\cite{park2021conceptual} further shows that the pricing rule of most DEXes gives rise to intrinsically profitable front-running opportunities. By analyzing the market participants' private incentives to prevent front-running, we build on these earlier works and, in particular, study sandwich attacks, as they influence the welfare of traders and liquidity providers. 

The prevalence of front-running attacks on DEXes is first quantified by Qin et al.~\cite{qin2021quantifying}. Zhou et al.~\cite{zhou2020highfrequency} study sandwich attacks both analytically and empirically. Both demonstrate the risk stemming from front-running attacks on DEXes. Our work, on the other hand, focuses on whether the market participant's private incentives are disruptive to the adoption of market designs preventing front-running attacks and the associated risks. 

Recently, many suggestions for DEX front-running prevention schemes have emerged. For a comprehensive overview, we refer the reader to Heimbach and Wattenhofer~\cite{heimbach2022sok}. Their work compares state-of-the-art prevention mechanisms and finds that current schemes do not meet the blockchain ecosystem's requirements. 

We summarise the core ideas behind the most relevant suggestions in the following. The simplest schemes, tune the transaction parameters to prevent specific attacks on specific protocols~\cite{zhou2021a2mm,heimbach2022eliminating}. Further, several suggestions propose that transactions are sent to a trusted third party that is put in charge of ordering the transactions fairly~\cite{2021flash,2021eden,2022openmev,2022gnosis,2022cowswap,bentov2019tesseract,stathakopoulou2021adding,2022ata,2022secretswap}. A parallel line of work, instead of relying on a single entity, trust a generally permissoned committee to order the transactions in a fair manner~\cite{kelkar2020order,baird2016swirlds,kursawe2020wendy,zhang2020byzantine,kelkar2021order,kelkar2021themis,cachin2021quick,reiter1994securely,miller2016honey,asayag2018fair,orda2021enforcing,zhang2022flash,constantinescu2023fair,momeni2023fairblock} --- preventing front-running. Finally, several schemes set the order of transactions by first having users commit to their transactions on-chain and then only revealing the transaction contents later in a second phase~\cite{tatabitovska2021mitigation,breidenbach2018enter,doweck2020multi}. All of these schemes, thus, aim to preserve the privacy of the transaction contents until an execution order is agreed upon. In this work, instead of assessing or designing prevention mechanisms, we study a market with an ideal prevention mechanism to analyze whether the private incentives would steer market participants to accept such a market design.

Budisch et al.~\cite{budisch2019theory} examine the incentives of exchanges to embrace market design innovations that eliminate latency arbitrage and HFT trading. Their work finds that adoption failures arise in traditional limit order book exchanges. In particular, the divergence of private and social incentives hinders the market from accepting new market designs. We study the incentives of liquidity providers in DEXes and show that their private incentives generally align with the system's social interests: demonstrating that liquidity providers can be incentivized to adopt new market designs preventing front-running attacks. 

%-------------------------------------------------------------------------------
\section{Preliminaries}
\label{sec: Preliminaries}
%-------------------------------------------------------------------------------
In the following, we detail the trading mechanism of the biggest DEXes and introduce the sandwich attack specifics.

\subsection{Automated Market Maker}
As its name suggests, trade execution on \emph{automated market makers (AMMs)} is automatic, and the price is controlled by an algorithm with liquidity being supplied by individual liquidity providers rather than brokers or market makers. A host of AMM variants exist, each with its specific pricing mechanism. We focus on the most widely adopted subclass of AMMs: \emph{constant product market makers (CPMMs)}~\cite{2021defillama}. For each tradeable cryptocurrency pair, the CPMM stores assets of both cryptocurrencies in a liquidity pool. The CPMM then guarantees that the product between the amounts of the two reserved pool currencies stays constant. This property ensures that the price for swapping between these pairs mimics the behavior of a demand curve of a normal good. Both Uniswap and Sushiswap, two of the biggest DEXes, employ the CPMM for pricing. The original CPMM design, as deployed by Uniswap V2~\cite{adams2020uniswap} and Sushiswap~\cite{2021sushiswap}, utilizes the same liquidity for the pool's entire price range. Consider a pool $X\rightleftharpoons Y$ between $X$-tokens and $Y$-tokens with respective reserves $x$ and $y$. Then the pool's marginal price indicating the pool's current price for $X$-token in terms of $Y$-token is $P = y/x$~\cite{adams2020uniswap}. Further, the pool's liquidity is defined as $L=\sqrt{x\cdot y}$. This liquidity needs to support trading along the entire price range $(0,\infty)$ in the original CPMM design.

In the newest Uniswap design (V3)~\cite{adams2021uniswap}, liquidity providers choose the price range $[P_a,P_b]$ for which they provide liquidity. This concentration of liquidity is intended to increase capital efficiency, as the liquidity only needs to support trade execution in the corresponding price range. Liquidity providers can only choose from a discrete set of price range boundaries that are defined by the pool's initialized \emph{ticks}. Between each pair of initialized ticks, the CPMM only needs to maintain enough reserves to support trading between the price boundaries. One can simulate a constant product pool with adjusted larger reserves, referred to as \emph{virtual reserves}, between any pair of neighboring initialized ticks.

For a price range $[P_a,P_b]$ between two neighboring initialized ticks, the \emph{liquidity} inside the tick is given by $L$ and the \emph{marginal price} is $P$. The CPMM ensures that the constant product of the virtual reserves $x$ and $y$ stays constant, i.e.,
$x\cdot y=k =L^2$,
where $k$ is the constant product of the reserves in the considered price interval. As on Uniswap V2, the marginal price is $P= y/x$ and the liquidity is $L= \sqrt{x \cdot y}$~\cite{adams2021uniswap}. Thus, the virtual reserves are then given by 
$x = {L}/{\sqrt{P}} $ and $ y = L \cdot \sqrt{P}$. For the sake of simplicity, we focus on trading between two neighboring initialized ticks and refer to virtual reserves simply as reserves in this work.\footnote{As our analysis focuses on a time frame where the fair market price between $X$- and $Y$-tokens remains constant (cf. Section~\ref{sec:model}), trading is also likely to occur within a small price range and thus will likely remain within one tick. To cover trading across ticks, one can reapply our analysis. Note that trading within a tick on Uniswap V3 is mostly the same as on Uniswap V2.}

The exchange rate received by traders is dependent on their trade size and the number of tokens reserved in the liquidity pool. Consider, again, a liquidity pool between $X$-token and $Y$-token, $X\rightleftharpoons Y$. We denote the respective initial (virtual) reserves prior to any trading as $x$ and $y$, respectively, and use $\xi$ and $\upsilon$ for the reserves of $X$ and $Y$ at any time during the trading process. Thus, if a trader adds an infinitesimal amount $d\xi$ to the pool, the following amount $d\upsilon$ of $Y$ is extracted where
\begin{equation*}
   d\upsilon=-\dfrac{xy}{\xi^2} d\xi.
\end{equation*}
This expression follows directly from the constant product property. Note that the sign convention we choose is relative to the pool, i.e., $\Delta \upsilon<0$ for a trader buying $Y$-tokens. Further, observe that the price per $Y$-token increases with the input amount. Thus, traders have to pay more per desired token the larger their trade is, resulting in \emph{expected slippage} which is the difference between the pool's marginal price and the actual price received by the trader. Note that the expected slippage is lower in more liquid pools, i.e., those with larger stored reserves. 

From the infinitesimal price, it follows that a trader wishing to sell $\delta _x$ $X$-tokens will receive $\delta _y$ $Y$-tokens, where
\begin{align*}
    \delta _{y}&= -\int_ x ^{x+ (1-f)\delta_x } \dfrac{-x \cdot y}{\xi ^2} d\xi \\ &=   y - \frac{x\cdot y}{x+(1-f)\delta_{x}} = \frac{y(1-f )\delta_{x}}{x+(1-f)\delta_{x}}, %\label{eq:exchange} 
\end{align*}
and $f$ is the transaction fee which is charged relative to the input amount $\delta_x$ and is distributed pro-rata to the tick's liquidity providers. The sign in front of the integral is negative as the trader receives the $Y$-tokens extracted from the pool. Note that the $(1-f)\delta_x$ in the upper integral bound corresponds to the amount of $X$ added to the pool after deduction of the transaction fee. Thus, post-execution the reserves of $X$ and $Y$ will be $x+(1-f)\delta_x$ $X$-token and $y-\delta_y$ $Y$-tokens.

The time at which a trade executes is unclear to the trader, as their transactions will only be confirmed upon block inclusion. In the meantime, other transactions changing the pool's state might occur. The change in the pool's state introduces a difference between the trader's expected price at the time of submission and the actual price at the time of execution. This price change is known as \emph{unexpected slippage}. To ensure the price of the transaction does not deviate significantly from the expectation, traders specify a \emph{slippage tolerance}, indicating the maximum unexpected price movement they are willing to accept. A trade expecting $\delta _{y}$ $Y$-tokens at the time of transaction submission will only execute if it receives no less than $(1-s)\delta _{y}$ $Y$-tokens, where $s$ is the slippage tolerance. Typical slippage tolerances are $s<0.03$~\cite{Wang2022impact}. 

\subsection{Sandwich Attacks}

A too small slippage tolerance results in frequent transaction failure. However, the slippage tolerance also gives an opening for sandwich attacks. On Ethereum, users broadcast their transactions to the network. The transaction waits in the mempool until it is included in a block by a validator. During this time, the transaction is visible to predatory trading bots and runs the risk of being front- and back-run as part of a sandwich attack. Predatory trading bots scan the mempool's inflowing transaction stream searching for profitable sandwich attack opportunities.

As validators\footnote{Note that currently most blocks are built through the proposer-builder separation scheme, where block building is outsourced to the specialized builders~\cite{PBSEthereum2023}. The same reasoning we detail for the validator also applies to these builders.} control the ordering in a block, sandwich attackers can provide validators with the necessary (financial) incentives to achieve their desired transaction ordering. In fact, \emph{front-running-as-a-service} schemes, such as Flashbots~\cite{2021flash} and Eden network~\cite{2021eden}, facilitate this interaction between sandwich attackers and validators. On the other hand, these services can also be used for front-running prevention, but users must deliberately seek them out rather than being truly incorporated into the market design.

\begin{figure*}[t]
\centering
  \begin{subfigure}[t]{0.48\linewidth}
  \centering
    \definecolor{color1}{HTML}{9C7CA5}
\definecolor{color2}{HTML}{A71D31}
\definecolor{color3}{HTML}{9DB17C}
\definecolor{color4}{HTML}{036016}
\begin{tikzpicture}[scale=2]
  \fill [color4!30, opacity=0.6, domain=0.71:1.41, variable=\x]
      (0.71,3.4)
      -- plot ({\x}, {2/\x})
      -- (1.41, 3.4)
      -- cycle;
     \fill [color3!30, opacity=0.6, domain=0.71:1.41, variable=\x]
      (0., 2.8)
      -- plot ({\x}, {2/\x})
      -- (0, 1.41)
      -- cycle; 
  \draw[-stealth,thick] (0, 0) -- (3.5, 0) node[midway, below,sloped] {\small $X$ reserves};
  \draw[-stealth,thick] (0, 0) -- (0, 3.5) node[midway, above,sloped] {\small $Y$ reserves};

  \draw[scale=1,  line width=0.4mm, domain=0.588:3.4, smooth, variable=\x, gray] plot ({\x}, {2/\x});
  %\node at (4,15)[]{\Large price curve: $x\cdot y=k$};
  %\node at (3.3,6)[anchor = west]{\Large (90 $A$, 180 $B$) };
  %\node at (0.71,2.8)[anchor = west]{\small ($x$,$y$) };
   \node at (0.71,2.8)[circle,fill,inner sep=1.3pt](A) {} ;
   %\node at (1.65,1.333)[anchor = west]{\small ($x^{\tilde{v}}$, $y^{\tilde{v}}$) };
   \node at (1.41,1.41)[circle,fill,inner sep=1.3pt](C) {} ;
   
   \draw [-stealth,  line width=0.4mm] (A) to [bend right]  node[midway,left, inner sep = 7pt] {$T$}(C) ;
   \draw[thick,stealth-stealth] (0.71, 3.05) -- (1.41,3.05) node[midway,above] {$\delta _{x}$};
   
   \draw[thick,stealth-stealth] (0.45, 2.8) -- (0.45,1.41) node[midway,left] {$\delta_{y}$};

\end{tikzpicture}
    \caption{The execution of transaction $T$ without a sandwich attack. The transaction $T$ receives the $Y$-assets at the expected price.} \label{fig:trade0}
\end{subfigure}
    \hfill
  \begin{subfigure}[t]{0.48\linewidth}
  \centering
    \definecolor{color1}{HTML}{EF767A}
\definecolor{color2}{HTML}{456990}
\definecolor{color3}{HTML}{9DB17C}
\definecolor{color4}{HTML}{036016}

\begin{tikzpicture}[scale=2]
    % \fill [color1!20, domain=1.088:1.594, variable=\x]
    %   (1.088, 0)
    %   -- plot ({\x}, {2/\x})
    %   -- (1.594, 0)
    %   -- cycle;
      
    % \fill [color1!20, domain=0.71:0.896, variable=\x]
    %   (0.71, 0)
    %   -- plot ({\x}, {2/\x})
    %   -- (0.896, 0)
    %   -- cycle;

    % \fill [color2!20, domain=0.71:0.896, variable=\x]
    %   (3,2.8)
    %   -- plot ({\x}, {2/\x})
    %   -- (3, 2.219)
    %   -- cycle;
      
    % \fill [color2!20, domain=1.088:1.594, variable=\x]
    %   (3,1.827)
    %   -- plot ({\x}, {2/\x})
    %   -- (3, 1.247)
    %   -- cycle;
    
     \fill [color3!30, opacity=0.6, domain=0.896:1.594, variable=\x]
      (0,2.219)
      -- plot ({\x}, {2/\x})
      -- (0, 1.247)
      -- cycle;
      
     \fill [color4!30, opacity=0.6, domain=0.896:1.594, variable=\x]
      (0.896,3.4)
      -- plot ({\x}, {2/\x})
      -- (1.594, 3.4)
      -- cycle;
     \fill [color1!30, opacity=0.6, domain=1.088:1.594, variable=\x]
      (1.088, 0)
      -- plot ({\x}, {2/\x})
      -- (1.594, 0)
      -- cycle;
      
    \fill [color1!30, opacity=0.6,domain=0.71:0.896, variable=\x]
      (0.71, 0)
      -- plot ({\x}, {2/\x})
      -- (0.896, 0)
      -- cycle;

    \fill [color2!30, opacity=0.6, domain=0.71:0.896, variable=\x]
      (3.4,2.8)
      -- plot ({\x}, {2/\x})
      -- (3.4, 2.219)
      -- cycle;
      
    \fill [color2!30, opacity=0.6,domain=1.088:1.594, variable=\x]
      (3.4,1.827)
      -- plot ({\x}, {2/\x})
      -- (3.4, 1.247)
      -- cycle;

  \draw[-stealth, thick] (0, 0) -- (3.5, 0) node[midway, below,sloped] {\small $X$ reserves};
  \draw[-stealth, thick] (0, 0) -- (0, 3.5) node[midway, above,sloped] {\small $Y$ reserves};
  \draw[scale=1,line width=0.4mm, domain=0.588:3.4, smooth, variable=\x, gray] plot ({\x}, {2/\x});
  %\node at (4,15)[]{\Large price curve: $x\cdot y=k$};
  %\node at (3.3,6)[anchor = west]{\Large (90 $A$, 180 $B$) };
  %\node at (0.85,2.5)[anchor = west]{\small ($x$,$y$) };
   \node at (0.71,2.8)[circle,fill,inner sep=1.3pt](A) {} ;
   %\node at (1.2,1.43)[anchor = east]{\small ($x^{\tilde{v}}$, $y^{\tilde{v}}$) };
   \node at (0.896,2.219)[circle,fill,inner sep=1.3pt](C) {} ;
   
   \node at (1.594,1.247)[circle,fill,inner sep=1.3pt](B) {} ;
   \node at (1.088,1.827)[circle,fill,inner sep=1.3pt](D) {} ;
   \draw [-stealth,line width=0.4mm] (A) to [bend right] node[midway,left] {$A_F$} (C) ;
   \draw [-stealth,line width=0.4mm] (C) to [bend right] node[pos =0.58,left,inner sep=6pt] {$T$}(B) ;
   \draw [-stealth, line width=0.4mm] (B) to [bend right] node[pos = 0.4,right,inner sep=6pt] {$A_B$}(D) ;

   \draw[thick,stealth-stealth] (3, 1.247) -- (3,1.827) node[midway,right] {$a_y$};

   \draw[thick,stealth-stealth] (3, 2.219) -- (3,2.8) node[midway,right] {$a_y$};

   \draw[thick,stealth-stealth] (0.71, 0.5) -- (0.896,0.5) node[midway,above] {$a_x^{\text{in}}$};
   
   \draw[thick,stealth-stealth] (1.088, 0.5) -- (1.594,0.5) node[midway,above] {$a_x^{\text{out}}$};
   
   \draw[thick,stealth-stealth] (0.896, 3.05) -- (1.594,3.05) node[midway,above] {$\delta _{x}$};
   
   \draw[thick,stealth-stealth] (0.45, 2.219) -- (0.45,1.247) node[midway,left] {$\tilde{\delta} _{y}$};
   
\end{tikzpicture}
    \caption{The execution of transaction $T$ with a sandwich attack. Transaction $T$ is first front-run by transaction $A_F$ and then back-run by transaction $A_B$.} \label{fig:trade1}
  \end{subfigure}\vspace{-5pt}
  
\caption{Execution of victim transaction $T$ in pool $X \rightleftharpoons Y$. without (cf. Figure~\ref{fig:trade0}) and with (cf. Figure~\ref{fig:trade1}) sandwich attack. In the presence of an attack the trader receives fewer Y-tokens $\tilde{\delta}_y<\delta_y$ while the attacker makes a profit, i.e., $a_x^\text{in}<a_x^\text{out}$.} \label{fig:trade}
\end{figure*}

A sandwich attack involves the attacker front-running the victim's transaction, exchanging $X$-token for $Y$-token in transaction $A_F$. The attacker's front-running transaction purchases the same asset as the victim: $Y$-token. Thereby, the attacker drives up the asset $Y$'s price. The following victim transaction $T$ then buys $Y$-token at a higher price and further inflates $Y$'s price. To conclude the attack, the attacker back-runs the victim's transaction, selling $Y$-assets at the inflated price with transaction $A_B$. 

To provide a conceptual understanding of sandwich attacks, we visualize a victim's trade $T$ without a sandwich attack in Figure~\ref{fig:trade0}. Figure~\ref{fig:trade1} then shows how a sandwich attack alters the transaction $T$. We observe that without the sandwich attack, the victim expects a greater output $\delta _y$ (cf. Figure~\ref{fig:trade0}) than the output $\tilde{\delta}_y$ it receives in the presence of a sandwich attack (cf. Figure~\ref{fig:trade1}). The attacker's front-running inflates $Y$-asset's price. Further, we observe that the attacker's output $a_x^{\text{out}}$ of the back-running transaction $A_B$ exceeds the attacker's input $a_x^{\text{in}}$ (cf. Figure~\ref{fig:trade1}). The difference $a_x^{\text{out}}-a_x^{\text{in}}$ presents the attacker's profit, as the attacker's input $a_y$ to transaction $A_B$ is the output of transaction $A_F$. 

Lastly, we note that at first glance liquidity providers appear to benefit from sandwich attacks as they lead to increased trading volume, and therefore, collected fees. However, traders aware of this threat could reduce their trading activity, as they receive a worse price than the market price if they fall victim to the attack. We will study this interplay by analyzing the effects of sandwich attacks on the utility of both traders and liquidity providers.

%-------------------------------------------------------------------------------
\section{Model}\label{sec:model}
%-------------------------------------------------------------------------------
We model a system with two liquidity pools $\text{Pool}_N$ and $\text{Pool}_W$. Both pools facilitate exchanges for the same cryptocurrency pair: $X\rightleftharpoons Y$. While a scheme to prevent sandwich attacks is implemented in $\text{Pool}_N$, sandwich attacks are common practice in $\text{Pool}_W$. With our model, we will study whether DEX participants are able to self-regulate and adopt a DEX with front-running prevention in place.

Our model has four types of players: (sophisticated and retail) traders, liquidity providers, sandwich attackers, and price arbitrageurs. Traders and liquidity providers strive to maximize their personal utility (cf. Section~\ref{sec:traderutility} and Section~\ref{sec:lputility}) across two liquidity pools $\text{Pool}_N$ and $\text{Pool}_W$. In maximizing their utilities, sophisticated traders and liquidity providers account for the effects of sandwich attacks and price arbitrageurs. Our model also captures the effects of less sophisticated traders who are oblivious to sandwich attacks. We will call this group retail traders. Further, for liquidity providers, we will also consider the consequences of them being inert. 

Without sandwich attacks, trades in $\text{Pool}_N$ execute at the expected price.\footnote{While it is possible for there to be several trades in a single block, we can assume them to only amount to natural price fluctuations. In the time frame of a block, they can be assumed to be negligible~\cite{heimbach2022eliminating}.} In $\text{Pool}_W$, on the other hand, sandwich attack bots make an attack whenever it is profitable (cf. Section~\ref{sec:sandwichattackmodel}). We denote the fraction of the total liquidity placed in $\text{Pool}_N$ by $p$. Thus, the reserves in $\text{Pool}_N$ are $x_N=p \cdot x$ $X$-tokens and $y_N =p \cdot y$ $Y$-tokens, and $x_W =(1-p) x$ $X$-tokens and $y_W =(1-p) y$ $Y$-tokens in $\text{Pool}_W$. Given the price $P_{X\rightarrow Y} $ of $X$-token, we have
$$P_{X\rightarrow Y} = \dfrac{y}{x}=\dfrac{y_N}{x_N}= \dfrac{p \cdot y}{p\cdot x}=\dfrac{y_W}{x_W}= \dfrac{(1-p)y}{(1-p)x}.$$

We emphasize that the social incentives of our system are to completely adopt $\text{Pool}_N$, the pool without front-running. In the presence of sandwich attacks in $\text{Pool}_W$, the trades of ordinary traders do not execute at the effective market price but rather at an unfavorable rate. Further, we purposefully exclude incentives of sandwich attackers and price arbitrageurs when discussing the system's incentives. Including their incentives would turn the game into a zero-sum game. Thus, in the presence of profitable sandwich attacks and price arbitrages, the remaining market participants (traders and liquidity providers) collectively lose money. 

Further note that throughout, we assume that the transaction fee $f$ ($0<f<1$) is identical in both pools. Additionally, we disregard the \emph{gas fee}, the fee paid to validators for block inclusion on the Ethereum blockchain, for all players in our analysis. The gas fee would add a fixed cost to every trade and liquidity movement. However, for the sake of simplicity and as the gas fee is not part of the CPMM market mechanism itself, we assume it to be zero.

\subsection{Sandwich Attackers}\label{sec:sandwichattackmodel}
Sandwich attackers observe the inflowing transactions in $\text{Pool}_W$. Upon noticing a trade entering the mempool of $\text{Pool}_W$, they compute the maximal input for the sandwich attack and assess the attack's profitability. The maximal input infers the maximal acceptable price movement on the trader, such that the trade still executes. Attackers conduct any such profitable attack. We find the maximal input of a sandwich attack and study their profitability in Section~\ref{sec:sandwichattackprofit}. 

The victim submits an order $T_W$ wishing to exchange $\delta_{x,W}>0$ $X$-tokens in $\text{Pool}_W$ for $Y$-tokens and sets a slippage tolerance $s$. When submitting the trade $T_W$, the victim is estimated to receive $\delta_{y,W}$ $Y$-tokens, i.e., the number of tokens the victim would receive if no other trade is executed beforehand. On the other hand, when a sandwich attack occurs, the attacker front-runs the victim with transaction $A_F$ exchanging $a_{x}^{\text{in}}>0$ $X$-tokens for $a_y$ $Y$-tokens. Now the victim's transaction executes at a worse price. To finish the attack, the attacker exchanges ${a_y}$ $Y$-tokens for $a_{x}^{\text{out}}$ $X$-tokens in the back-running attack transaction $A_B$. 

We define the attacker's utility as their profit:
\vspace{-3pt}
\begin{definition}\label{def:profit}
    The attacker's utility $U^A(\delta_{x,W},f, s, p, x,y)$ is given by $$a_{x}^{\text{out}}(\delta_{x,W},f, s, p, x,y)-a_{x}^{\text{in}}(\delta_{x,W},f, s, p, x,y).$$
    Here, $a_{x}^{\text{in}}(\delta_{x,W},f, s, p, x,y)$ is the input of the front-running transaction and $a_{x}^{\text{out}}(\delta_{x,W},f, s, p, x,y)$ is output of the back-running transaction. 
\end{definition}\vspace{-3pt}

We will assume that if a profitable sandwich attack exists, it executes successfully. A bot must have access to the necessary funds and achieve its desired transaction ordering which can be accomplished through \emph{front-running-as-a-service} platforms such as Flashbots~\cite{2021flash}. These services further guarantee their users that a transaction will only be included in a block if it executed successfully.  Therefore, it is reasonable to assume that profitable sandwich attacks execute successfully. 

\subsection{Price Arbitrageurs}
We consider a time window, during which the external market price between the pools' two cryptocurrencies is constant. Then, price arbitrageurs ensure that the pool's price returns to $P_{X\rightarrow Y} $ after every trade sequence (either an individual victim transaction in $\text{Pool}_N$ or a victim transaction together with a sandwich attack in $\text{Pool}_W$). Thus, price arbitrageurs balance the market after any set of trades to reflect the fair market price. Letting the pool return to its initial state allows us to study the system analytically in the presence of an infinitely long trade flow as opposed to a fixed set of trades.

\subsection{Traders}\label{sec:traderutility}
Our game captures a continuous stream of trade orders from two types of traders: sophisticated and retail. Sophisticated traders are aware and adjust to the presence of sandwich attacks, while retail traders are oblivious to the presence of these attacks, i.e., they trade in the pool as if there were no sandwich attacks. In the continuous stream of trade orders, a proportion $(1-\omega)$ of orders originate from sophisticated traders and a proportion $\omega$ of orders from retail traders. 

In $\text{Pool}_N$, where there are no sandwich attacks, the tokens received by traders equal the expected trade output. On the other hand, in $\text{Pool}_W$ traders experience sandwich attacks which reduce the expected output. Sophisticated traders account for these attacks, while retail traders do not. \footnote{Sophisticated traders assume there to be a sandwich attack for every transaction in $\text{Pool}_W$. As sandwich attacks only execute when they are profitable, there is not always a sandwich attack. However, this is only the case for small transactions (cf. Section~\ref{sec:sandwichattackprofit}) and unrealistic parameter configurations (cf. Section~\ref{sec:equilibria}), and it is, therefore, negligible.}

All traders wish to sell $X$-tokens for $Y$-tokens, as they have a personal use for $Y$-tokens. Thus, the sophisticated trader's \emph{strategy space} is $$S^S = \{(\delta_{x,N},\delta_{x,W}) \vert \delta_{x,N},\delta_{x,W}\in \mathbb{R}^{\geq 0}\},$$ 
while the retail trader's \emph{strategy space} is
$$S^T = \{(\Delta_{x,N},\Delta_{x,W}) \vert \Delta_{x,N},\Delta_{x,W}\in \mathbb{R}^{\geq 0}\},$$
where $\delta_{x,N}$ and $\Delta_{x,N}$ are the respective trade input sizes in $\text{Pool}_N$, whereas $\delta_{x,W}$ and $\Delta_{x,W}$ are the corresponding trade inputs in $\text{Pool}_W$. 

Traders set their trade sizes across both pools to maximize their personal benefit. As the traders have a personal use for $Y$-tokens, they associate a relative benefit $\alpha>0$ with $Y$-tokens. The private benefit associated with the number of $Y$-tokens a trader buys, $\delta_{y,\bullet}$, is thus given by  $(1+\alpha)\delta_{y,\bullet}$ in $\text{Pool}_{\bullet}$. In addition to the benefits traders obtain from the received $Y$-tokens, they associate a cost with the trade's input, which is given by $P_{X\rightarrow Y}\delta_{x,\bullet}$. Here, $\delta_{x,\bullet}$ is the trade input in $X$-tokens, and $P_{X\rightarrow Y}$ is the fair exchange rate from $X$-tokens to $Y$-tokens. Combining the trader benefit and cost in both pools, we obtain their utility for the sophisticated traders in Definition~\ref{def:utilitysophtrader} and for the retail trader in Definition~\ref{def:utilityretailtrader}. Notice that the important difference between the two utilities below is that the sophisticated trader takes the change in the output amount in $\text{Pool}_W$ as a consequence of sandwich attacks into account, while the retail trader does not. We indicate this with the lack of the argument $s$ in Definition~\ref{def:utilityretailtrader}. Further, note that the retail trader's utility $U^R$ can be seen as the \textit{expected} utility, i.e., what the retail trader expects and thus behaves according to. The \textit{realized} utility of the retail trader is lower in the presence of attacks and equivalent to that of the sophisticated trader.

\begin{definition}\label{def:utilitysophtrader}
    The sophisticated trader's utility $U^S(\delta_{x,N},\delta_{x,W},$ $\alpha,f, s, p, x,y)$ for a trade with input $\delta_{x,N}\geq 0$ in $\text{Pool}_N$ and input $\delta_{x,W}\geq 0$ in $\text{Pool}_W$ is given by 
    \begin{align*}
        &(1+\alpha)\delta_{y,N}(f,p,x,y)-\tfrac{y}{x}\delta_{x,N}\\ &+(1+\alpha)\delta_{y,W}(f,p,x,y,s)-\tfrac{y}{x}\delta_{x,W},
    \end{align*}
 
    Here, $\delta_{y,N}(f,p,x,y)$ and $\delta_{y,W}(f,p,x,y,s)$ are the outputs of the trade in each pool.
\end{definition}

\begin{definition}\label{def:utilityretailtrader}
    The retail trader's utility $U^R(\Delta_{x,N},\Delta_{x,W},\alpha,f, p, $ $x,y)$ for a trade with input $\Delta_{x,N}\geq 0$ in $\text{Pool}_N$ and input $\Delta_{x,W}\geq 0$ in $\text{Pool}_W$ is given by 
    \begin{align*}
        &(1+\alpha)\Delta_{y,N}(f,p,x,y)-\tfrac{y}{x}\Delta_{x,N}\\ &+(1+\alpha)\Delta_{y,W}(f,p,x,y)-\tfrac{y}{x}\Delta_{x,W},
    \end{align*}
 
    Here, $\Delta_{y,N}(f,p,x,y)$ and $\Delta_{y,W}(f,p,x,y)$ are the outputs of the trade in each pool.
\end{definition}

Observe that in this framework, trades execute across both pools to maximize the respective trader's utility. Given a distribution on the relative benefit $\alpha$ and slippage tolerance $s$, the trading volume in either pool depends on the pool's reserve, transaction fee, and slippage tolerance. Throughout we assume all traders have the same relative benefit $\alpha$ and slippage tolerance $s$. Later we will also consider distribution on $\alpha$ to capture non-uniformity among traders. Further note that while we focus on trades from $X$ to $Y$, the analysis applies directly in the opposite direction by symmetry.

\subsection{Liquidity Providers}\label{sec:lputility}
Liquidity providers supply reserves to the two pools. Knowledge of the sophisticated and retail trader's utility is assumed for liquidity providers. Further, liquidity providers are aware of the behavior of sandwich attackers and price arbitrageurs. We consider the liquidity providers to be rational, i.e., they optimally place their liquidity across the pools such that they maximize their received fees. The system has $n \in \mathbb{N}$ liquidity providers. Both the number of liquidity providers and the system's total reserves are fixed during the time of this analysis. A liquidity provider $LP_i$ for $i \in [0,\dots,n-1]$ holds a proportion $l_i$ ($0< l_i \leq 1$) of the total liquidity $L=\sqrt{x \cdot y}$, where $x$ and $y$ are the system's total reserves. We note that $\sum _{i=0}^{n-1} l_i=1$. 

Liquidity provider $LP_i$'s \emph{strategy space} is given by all possible distributions of their liquidity across both pools: $$S^{LP}_i = \{(p_i l_i  L,(1-p_i)  l_i  L) \vert 0\leq p_i\leq 1\}.$$ More precisely, a liquidity provider $LP_i$ can choose the proportion $p_i$ of their liquidity in $\text{Pool}_N$. They automatically place the remaining proportion $1-p_i$ of their liquidity in $\text{Pool}_W$. Knowing the distribution of the remaining liquidity $(1-l_i)L$ across the pools and the behavior of the other market participants, the liquidity provider chooses the strategy that maximizes the received fees. 
We define the liquidity provider's utility as the earned fees: 
\begin{definition}\label{def:utilitylp}
    The utility $U^{\text{LP}} (f,\alpha,s,x,y,p_i,l_i)$ of liquidity provider $LP_i$ that places $p_i l_i L$ liquidity in $\text{Pool}_N$ and $(1-p_i) l_i L$  in $\text{Pool}_W$ represents the fees collected in both pools.
\end{definition}

Our game starts with an arbitrary initial liquidity distribution. One after the other, liquidity providers can change their personal liquidity distribution. The system is in a \emph{Nash equilibrium} if no liquidity provider can improve their utility by unilaterally changing their liquidity distribution (strategy). We will loosen the restriction on equilibria and also consider \emph{$\varepsilon$-equilibria}, where a liquidity provider only changes strategy if it increases their utility by a factor greater than $1+\varepsilon$ ($\varepsilon\geq 0$). We analyze the system with this relaxation on equilibria, as inert liquidity providers are unlikely to change strategies for infinitesimal utility increases due to the effort involved. This adjustment allows us to analyze whether the potential private benefits of liquidity providers suffice.

%-------------------------------------------------------------------------------
\section{Strategies}
%-------------------------------------------------------------------------------
The optimal strategies of sandwich attackers and price arbitrageurs are straightforward. Sandwich attackers always execute the largest possible profitable attack, i.e., the attack inferring the maximal acceptable price movement on the trader (cf. Section~\ref{sec:sandwichattackprofit}), and price arbitrageurs re-balance the market after every trade sequence.

Sophisticated and retail traders set their trade sizes across both pools optimally to maximize their utility, knowing the pools' liquidity, transaction fee, and, in the prior case, the potential presence of sandwich attacks (cf. Section~\ref{sec:tradesize}). Finally, liquidity providers distribute their liquidity to maximize the received fees. Liquidity providers account for the effects their decision to alter the liquidity distribution would have on the trading volume of the other market participants (cf. Section~\ref{sec:liquiditydistribution}). Note that the section's omitted proofs can be found in Appendix~\ref{app:strategies}.

\subsection{Sandwich Attack Profitability}\label{sec:sandwichattackprofit}
A sandwich attacker only executes an attack whenever it is profitable, i.e., when $U^A$ is positive (cf. Definition~\ref{def:profit}). We find that the sandwich attacker's profit for a front-running transaction of size $a^{\text{in}}_{x}$ can be calculated analytically in Lemma~\ref{lem:profit}. The expression is given in Appendix~\ref{app:strategies}.
\begin{restatable}[]{lemma}{profit}\label{lem:profit}
    The sandwich attacker's profit from an attack of size $a^{\text{in}}_{x} $ to the front-running transaction on a victim's transaction $\delta_{x,W}$  can be given analytically.
\end{restatable}

We will analyze the conditions under which profitable sandwich attacks exist. First, we determine a bound for the victim's trade size $\delta_{x,W}$ such that a profitable sandwich attack exists (cf. Lemma~\ref{lem:smalldelta}). From Lemma~\ref{lem:smalldelta} we can follow that a profitable attack only exists, if the victim's trade size $\delta_{x,W}$ exceeds a fee dependent threshold $$\delta_x^\text{min}=\frac{f(1-p)x}{(1-f)^2}.$$ Hence, only relatively large trades are prone to sandwich attacks.
\begin{restatable}[]{lemma}{smalldelta}\label{lem:smalldelta}
    A sandwich attack of size $a_x^\text{in}$ is only profitable if the trader's transaction size exceeds $$\frac{f((1-p)x+a_x^\text{in}(1-f))}{(1-f)^2}.$$
\end{restatable}

Next, we explore what limits the attacker's maximum profit to show that it is optimal for sandwich attackers to execute the attack with maximal input size, i.e., the attack that infers the maximal acceptable price movement, as dictated by the slippage tolerance of the trader. In Lemma~\ref{lem:front} we show that the attacker's maximal input $a^{s}_{x}$ for which a victim's transaction still executes can be calculated analytically.
\begin{restatable}[]{lemma}{front}\label{lem:front}
    The sandwich attacker's maximal input, $a^{s}_{x}$, for a transaction exchanging $\delta _{x,W}$ $X$-tokens with slippage tolerance $s$ such that the victim's trade still executes can be given analytically.
\end{restatable}

However, for very large slippage tolerances the size of the sandwich attack is limited. To see this we can consider the asymptotic behavior of Lemma~\ref{lem:profit} in the limit of very large attack sizes $a_x^\text{in}$: 
$\lim _{a_x^\text{in}\rightarrow \infty} U^A =\lim _{a_x^\text{in}\rightarrow \infty}-f a_x^\text{in} \rightarrow -\infty.$
We, thus, analyze whether the slippage tolerance or profitability limits the sandwich attack size and plot the sandwich attack size that achieves the maximum profit $U^A$ as a function of the victim's transaction size $\delta_{x,W}$ in Figure~\ref{fig:maxprofit0}. Figure~\ref{fig:maxprofit1} shows the sandwich attacker's maximal input $a^{s}_{x}$ as a function of the victim's transaction size $\delta_{x,W}$ for different slippage tolerances. Note the vast difference in scale, demonstrating that the sandwich attack size is clearly limited by the slippage tolerance. Thus, in realistic market configurations, the sandwich attackers always execute the attack with maximal possible input size $a^{s}_{x}$.

\begin{figure}[t]
    \centering
  \begin{subfigure}[t]{1\linewidth}
  \includegraphics[scale=1]{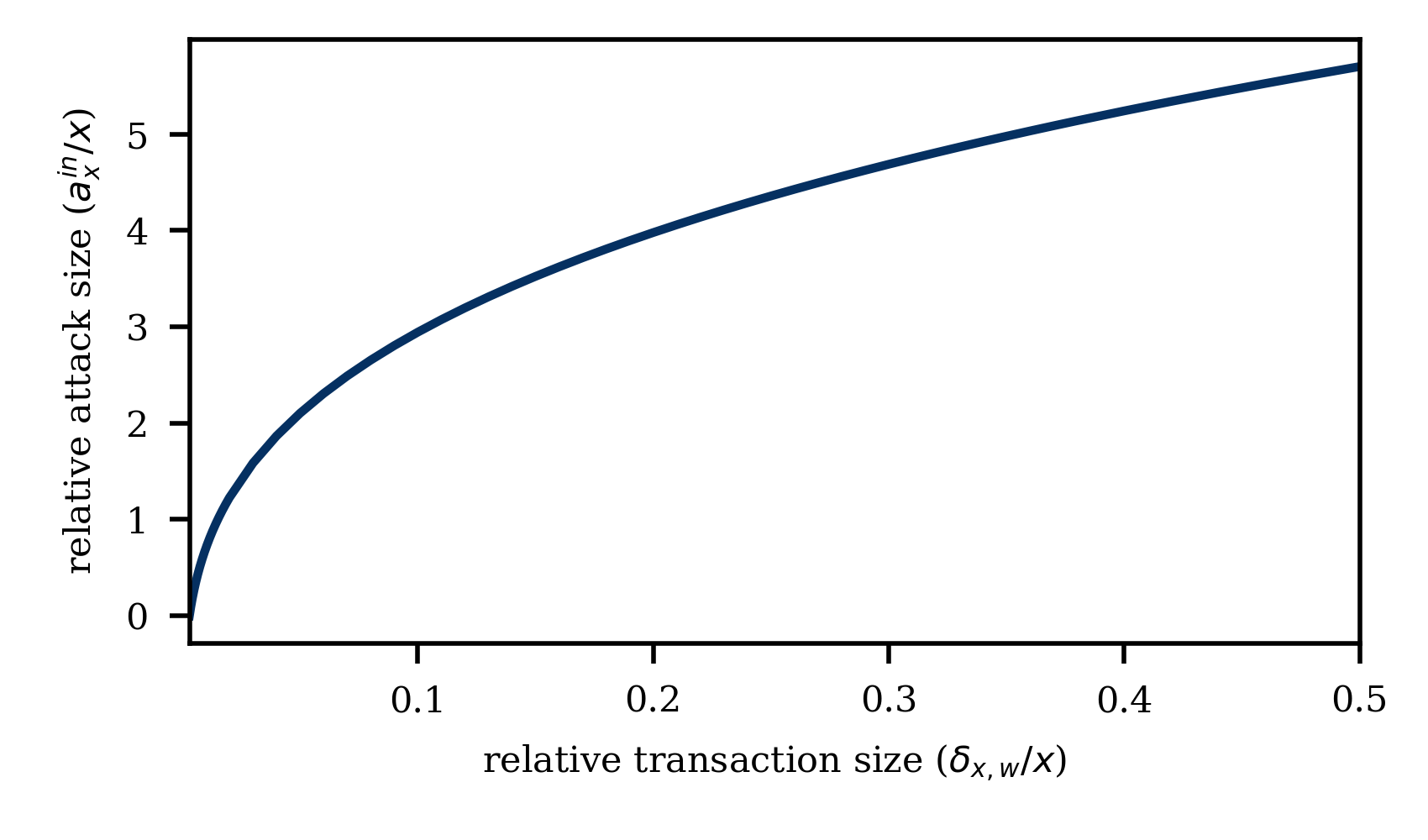}\vspace{-8pt}
  
    \caption{Attack size $a_x^\text{in}$ achieving max profit $U^A$ vs. the victim's trade size. $a_x^\text{in}$ is found by maximizing the attacker's profit w.r.t. $a_x^\text{in}$.} \label{fig:maxprofit0}
  \end{subfigure}%
  
  \begin{subfigure}[t]{1\linewidth}
    \includegraphics[scale=1]{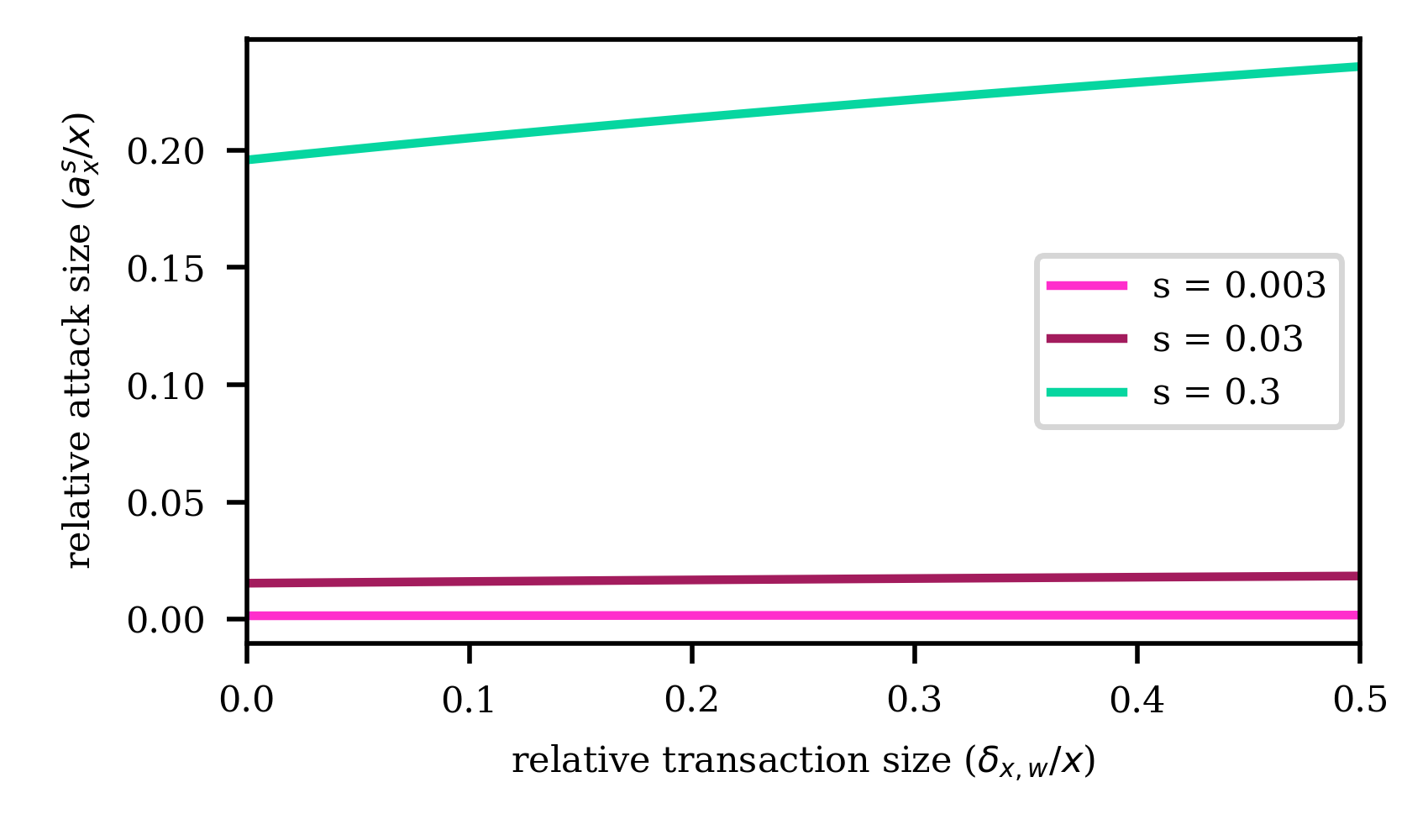}\vspace{-8pt}
    
    \caption{Maximum sandwich attack size dependent on the victim's transaction size and slippage tolerance.} \label{fig:maxprofit1}
  \end{subfigure}\vspace{-5pt}
  
    \caption{Limits on the sandwich attack size in terms of profitability (left) and slippage tolerance (right) for $f=0.3\%$. Note the vast difference in scale of the vertical axis, demonstrating that the attack is limited by the slippage tolerance.}
    \label{fig:maxprofit}
\end{figure}

\subsection{Trade Sizes}\label{sec:tradesize}

Sophisticated traders wish to maximize their utility $U^S(\delta_{x,N},\delta_{x,W},\alpha,f, s, p, x,y)$ (cf. Definition~\ref{def:utilitysophtrader}), i.e., the difference between the benefit from receiving $Y$-tokens and the trade's costs. The utility function accounts for sandwich attacks in $\text{Pool}_W$ and assumes that the transaction output is reduced by the slippage tolerance. Retail traders, on the other hand, wish to maximize their utility $U^R(\Delta_{x,N},\Delta_{x,W},\alpha,f, p, x,y)$ (cf. Definition~\ref{def:utilityretailtrader}) that does not account for the presence of sandwich attacks.

\Paragraph{Sophisticated Trader} We start with sophisticated traders and show in Lemma~\ref{lem:tradesizesoph} that the optimal transaction size, maximizing utility $U^R$, in $\text{Pool}_N$ ($\delta_{x,N} ^{\text{opt}}$) and $\text{Pool}_W$ ($\delta_{x,W} ^{\text{opt}}$) can be expressed analytically. Observe that the transaction size is proportional to the pool's reserves of $X$-token. Further, we can see that the effects of slippage tolerance on the trade size are identical to those of the transaction fee. Thus, the combination of transaction fee $f$ in $\text{Pool}_W$ and the trader's slippage tolerance $s$ is from the trader's perspective equivalent to a larger transaction fee equaling $f+s-f\cdot s$ in $\text{Pool}_N$. Therefore, the transaction size in $\text{Pool}_N$ is always larger than in $\text{Pool}_W$ at the same liquidity level, and we follow that the trader's utility is maximized for $p=1$. We also note that the optimal transaction size increases with $\alpha$ and decreases with the transaction fee $f$, as well as, where applicable, the slippage tolerance $s$.\vspace{-3pt}
\begin{restatable}[]{lemma}{tradesizesoph}\label{lem:tradesizesoph}
    A trade of size 
    $\delta_{x,N} ^{\text{opt}} = \max(0,p\cdot x (\sqrt{(1+\alpha)(1-f)} -1) /(1-f))$ 
    maximizes a sophisticated trader's utility $U^R $ 
    in $\text{Pool}_N$ and in $\text{Pool}_W$ the optimum is at
    $\delta_{x,W} ^{\text{opt}} = \max (0,(1-p) x( \sqrt{(1+\alpha)(1-s)(1-f)}-1) /(1-f))$.
\end{restatable}

With the help of Lemma~\ref{lem:tradesizesoph}, we can obtain bounds for relative benefit $\alpha$ as a function of the slippage tolerance $s$, such that sophisticated traders benefit from trading in $\text{Pool}_N$ and $\text{Pool}_W$. A trader executes a trade in $\text{Pool}_N$, as long as their $\alpha$ exceeds $$\alpha>\alpha^{\min}_N = \frac{f}{(1-f)}.$$ Notice that this bound only depends on the transaction fee $f$. In $\text{Pool}_W$, a sophisticated trader will only execute a trade if $$\alpha>\alpha^{\min}_W =\frac{f+s-s\cdot f}{((1-f)(1-s))}.$$

\Paragraph{Retail Trader}  In Lemma~\ref{lem:tradesizeretail}, we show the optimal trade sizes for the retail trader (i.e., those that maximize utility $U^R$) in $\text{Pool}_N$ ($\Delta_{x,N} ^{\text{opt}}$) and $\text{Pool}_W$ ($\Delta_{x,W} ^{\text{opt}}$) can also be expressed analytically. Note that the proof for Lemma~\ref{lem:tradesizeretail} is analogous to that of Lemma~\ref{lem:tradesizesoph}. In particular, in $\text{Pool}_N$ the behavior of the retail trader mirrors that of the sophisticated trader (i.e., $\Delta_{x,N} ^{\text{opt}} = \delta_{x,N} ^{\text{opt}} $). In $\text{Pool}_W$, on the other hand, the retail traders, who is oblivious to sandwich attacks, only adjust their trade size in response to differing levels of liquidity relative to $\text{Pool}_N$. Thus, the retail trader's trade size in $\text{Pool}_N$ is equivalent to that in $\text{Pool}_W$ at the same liquidity level, and we follow that the trader's utility is independent of $p$. Importantly, this is only due to the retail trader optimizing their \textit{expected} utility, which ignores the attacks. The \textit{realized} utility, which matches that of the sophisticated trader, is maximized when $p=1$.

\begin{restatable}[]{lemma}{tradesize}\label{lem:tradesizeretail}
    A trade of size 
    $\Delta_{x,N} ^{\text{opt}} = \max(0,p\cdot x (\sqrt{(1+\alpha)(1-f)} -1) /(1-f))$ 
    maximizes a retail trader's utility $U^R $ 
    in $\text{Pool}_N$ and in $\text{Pool}_W$ the optimum is at
    $\Delta_{x,W} ^{\text{opt}} = \max (0,(1-p) x( \sqrt{(1+\alpha)(1-f)}-1) /(1-f))$.
\end{restatable}

Further note that a retail trader executes a trade in $\text{Pool}_N$ and $\text{Pool}_W$, as long as their $\alpha$ exceeds $$\alpha>\alpha^{\min}_N = \frac{f}{(1-f)},$$
which is the same bound for the sophisticated trader in $\text{Pool}_N$.

\subsection{Liquidity Distribution}\label{sec:liquiditydistribution}
A liquidity provider's utility directly corresponds to the received fees (cf. Definition~\ref{def:utilitylp}). We will investigate how liquidity providers distribute their liquidity across the two pools', knowing that sophisticated and retail traders execute their respective optimal transactions. Recall, that in the continuous stream of trade orders that we study, a proportion $(1-\omega)$ stems from sophisticated traders, while a proportion $\omega$ stems from retail traders. Each of these trades is accompanied by a sandwich attack (whenever applicable), and price arbitrage. 

We quantify the system's total fees in Lemma~\ref{lem:totalfee} and that the total fees are proportional to $p$. If the fee gradient with respect to $p$ is zero, all liquidity distributions maximize the game's fees. Otherwise, the game's fees are maximized, either when all liquidity is in $\text{Pool}_N$ ($p=1$) or when all liquidity is in $\text{Pool}_W$ ($p=0$).\vspace{-3pt}
\begin{restatable}[]{lemma}{totalfee}\label{lem:totalfee}
    The total transaction fees $ F(f,\alpha,s,y,p,\omega)$ collected across both pools for retail and sophisticated traders with relative benefit $\alpha$ are proportional to $p$.
\end{restatable}
\begin{proof}
    We consider the fee revenue stemming from the order flow related to sophisticated and retail traders separately. 
    
    We start with sophisticated traders. There, we consider the following four intervals: 
    $$0<\alpha \leq \alpha^{\min}_N\text{, } \alpha^{\min}_N<\alpha<\alpha^{\min}_W\text{, }$$ 
    $$ \alpha>\alpha^{\min}_W \text{ and }U^A_S\leq 0\text{, }\alpha>\alpha^{\min}_W\text{ and }U^A_S\geq 0,$$
    where $U^A_S(\delta_{x,W},f, s, p, x,y)$ is the sandwich attacks' profitability for the trades from sophisticated traders (cf. Definition~\ref{def:profit}). 
    
    Following from Lemma~\ref{lem:tradesizesoph}, we conclude that no trades from sophisticated traders execute in either pool and, thereby, no fees collected for $\alpha \leq \alpha^{\min}_N $.
    
    We continue with the second interval, i.e., $\alpha^{\min}_N <\alpha \leq \alpha^{\min}_W $. Following from Lemma~\ref{lem:tradesizesoph} sophisticated traders exclusively execute trades in $\text{Pool}_N$ on this interval. The fees collected in $\text{Pool}_N$ for a transaction by a sophisticated trader with relative benefit $\alpha$ are
    \begin{align*}
    F_{N,S}(f,\alpha,s,y,p) 
    &=f\left( \delta_{x,N} ^{\text{opt}}\cdot \frac{y}{x} +  \tfrac{p\cdot y (1-f)\delta_{x,N} ^{\text{opt}} }{ p\cdot x +(1-f)\delta_{x,N} ^{\text{opt}}}\right)\\
    &=p \cdot y \cdot f\left(\tfrac{\alpha}{\sqrt{(1+\alpha)(1-f)}}-\tfrac{f}{1-f}\right)
    \end{align*}
    $Y$-tokens. In the previous, $\delta_{x,N} ^{\text{opt}} \cdot \frac{y}{x}$ is the sophisticated trader's transaction size in $Y$-tokens in $\text{Pool}_N$ and $$\tfrac{p\cdot y\delta_{x,N} ^{\text{opt}}}{p \cdot x +(1-f)\delta_{x,N} ^{\text{opt}}}$$ is the size of the price arbitrageur's transaction. 
    
    In the third interval, i.e., $\alpha>\alpha^{\min}_W $ and $U^A_S\leq 0$, sophisticated traders execute trades in both pools (cf. Lemma~\ref{lem:tradesizesoph}). However, there is no profitable sandwich attack, due to the small trade size in $\text{Pool}_W$ (cf. Lemma~\ref{lem:smalldelta}). The fees collected from the sophisticated trader order flow in $\text{Pool}_N$ are again given by $F_{N,S}(f,\alpha,s,y,p)$ but liquidity providers collect additional fees in $\text{Pool}_W$. The fees collected in $\text{Pool}_W$ for a transaction by a sophisticated trader with relative benefit $\alpha$ are given by 
    \begin{align*}
    &F_{W,S}^{U^A_S\leq 0}(f,\alpha,s,y,p) \\
    &=f\left( \delta_{x,W} ^{\text{opt}}\cdot \frac{y}{x} +  \frac{(1-p) y (1-f)\delta_{x,W} ^{\text{opt}} }{ (1-p)x +(1-f)\delta_{x,W} ^{\text{opt}}}\right)\\
    &=(1-p)  y \cdot f\left( \tfrac{(n_1(f,\alpha,s) -1)(1-f+n_1(f,\alpha,s) ))}{(1-f)n_1(f,\alpha,s) }\right)
    \end{align*}
    where $$ \delta_{x,W} ^{\text{opt}} \cdot \frac{y}{x}$$ is the trader's transaction size in $Y$-tokens in $\text{Pool}_W$ and $$\tfrac{p\cdot y \cdot \delta_{x,W} ^{\text{opt}}}{p \cdot x +(1-f)\delta_{x,W} ^{\text{opt}}}$$ is the size of the price arbitrageur's transaction that returns the pools to its initial state. Further 
    $$ n_1(f,\alpha,s)= \sqrt{(1+\alpha)(1-s)(1-f)}.$$
    
    Finally, we analyze the fourth interval, i.e., $\alpha>\alpha^{\min}_W$ and $U^A_S> 0$. In this interval, sophisticated trades execute in both pools and sandwich attacks execute in $\text{Pool}_W$. Thus, in addition to the fees $ F_{N,S}(f,\alpha,s,y,p)$ collected in $\text{Pool}_N$, we also consider the fees collected in $\text{Pool}_W$ from traders, price arbitrageurs, and sandwich attackers for the liquidity provider utility. In the presence of sandwich attacks, the fees from sophisticated flow in $\text{Pool}_W$ are given by 
    \begin{align*}
    &F_{W,S}^{U^A_S>0}(f,\alpha,s,y,p)\\
    &=\left( \left(\delta_{x,W} ^{\text{opt}}+a_x^s \right)\frac{y}{x} +  \tfrac{(1-p)y (1-f) \left(\delta_{x,W} ^{\text{opt}}+a_x^s \right)}{(1-p)x +(1-f) \left(\delta_{x,W} ^{\text{opt}}+a_x^s \right)}\right)\\
    &= \tfrac{(1-p)  y \cdot f(( n_1(f,\alpha,s)-3)+ n_2(f,\alpha,s))(( n_1(f,\alpha,s)+1-2f)+ n_2(f,\alpha,s))}{1(1-f)(( n_1(f,\alpha,s)-1)+ n_2(f,\alpha,s))} 
    \end{align*}
    where $\left(\delta_{x,W} ^{\text{opt}}+a_x^s \right)\frac{y}{x}$ combines the trader's transaction size in $\text{Pool}_W$ and the bot's front-running transaction size in $Y$-tokens (cf. Lemma~\ref{lem:front}). In the previous, 
    \begin{align*}
        &n_2(f,\alpha,s) = \sqrt{ \tfrac{2 n_1(f,\alpha,s)(1+s) + (1-s) (2+\alpha (1-s) -s)-(1+\alpha)f(1-s)^2}{1-s}}.
    \end{align*}
    Further, 
    $$\tfrac{(1-p)y (1-f) \left(\delta_{x,W} ^{\text{opt}}+a_x^s \right)}{\left((1-p)x +(1-f) \left(\delta_{x,W} ^{\text{opt}}+a_x^s \right)\right)}$$ is the combined size of the attacker's back-running transaction and the transaction that returns the pool to its initial state. 

    Next, we consider the fee revenue from the order flow associated with retail traders. Here, we consider three intervals 
    $$0<\alpha \leq \alpha^{\min}_N\text{, } \alpha>\alpha^{\min}_N \text{ and }U^A_R\leq 0\text{, }\alpha>\alpha^{\min}_N\text{ and }U^A_R\geq 0,$$
    where $U^A_R(\Delta_{x,W},f, s, p, x,y)$ is the sandwich attacks' profitability for the trades from retail traders (cf. Definition~\ref{def:profit}). 

    Following from Lemma~\ref{lem:tradesizeretail}, we conclude that no trades from retail traders execute in either pool and, thereby, no fees collected for $\alpha \leq \alpha^{\min}_N $.
    
    We continue with the second interval, i.e., $\alpha \geq \alpha^{\min}_N\text{ and }U^A_R\leq 0 $. In contrast, to sophisticated traders, for $\alpha \geq \alpha^{\min}_N $, retail traders trade in both pools (cf. Lemma~\ref{lem:tradesizeretail}). Further, in this interval, there is no profitable sandwich attack on the retail flow, due to the small trade size in $\text{Pool}_W$ (cf. Lemma~\ref{lem:smalldelta}). 
    The fees collected from the retail trader order flow in $\text{Pool}_N$ are given by $F_{N,R}(f,\alpha,s,y,p)$ which is equal to the fees collected in $\text{Pool}_N$ by sophisticated traders (cf. Lemma~\ref{lem:tradesizeretail}), i.e., 
    $$F_{N,R}(f,\alpha,s,y,p)=p \cdot y \cdot f\left(\tfrac{\alpha}{\sqrt{(1+\alpha)(1-f)}}-\tfrac{f}{1-f}\right)=:F_N.$$
    Further, liquidity providers collect additional fees in $\text{Pool}_W$. The fees collected in $\text{Pool}_W$ for a transaction by a retail trader with relative benefit $\alpha$ are given by 
    \begin{align*}
    &F_{W,R}^{U^A_R\leq 0}(f,\alpha,s,y,p) \\
    &=f\left( \Delta_{x,W} ^{\text{opt}}\cdot \frac{y}{x} +  \frac{(1-p) y (1-f)\Delta_{x,W} ^{\text{opt}} }{ (1-p)x +(1-f)\Delta_{x,W} ^{\text{opt}}}\right)\\
    &=(1-p)  y \cdot f\left( \tfrac{ f + f\cdot \alpha - \alpha \sqrt{(1-f)(1-\alpha)}}{(1-f)(1-\alpha)}\right).
    \end{align*}
   
    Finally, we analyze the third interval, i.e., $\alpha>\alpha^{\min}_N$ and $U^A_r> 0$, where retail traders trade in both pools and sandwich attacks are profitable in $\text{Pool}_W$. Here, in addition to the fees $ F_{N,R}(f,\alpha,s,y,p)$ collected in $\text{Pool}_N$, we also consider the fees collected in $\text{Pool}_W$ from traders, price arbitrageurs, and sandwich attackers for the liquidity provider utility. In the presence of sandwich attacks, the fees from retail flow in $\text{Pool}_W$ are given by 
    \begin{align*}
    &F_{W,R}^{U^A_R>0}(f,\alpha,s,y,p)\\
    &=\left( \left(\Delta_{x,W} ^{\text{opt}}+a_x^s \right)\frac{y}{x} +  \tfrac{(1-p)y (1-f) \left(\Delta_{x,W} ^{\text{opt}}+a_x^s \right)}{(1-p)x +(1-f) \left(\Delta_{x,W} ^{\text{opt}}+a_x^s \right)}\right)\\
    &= (1-p)  y \cdot f \left(\tfrac{n_3(f,\alpha,s)}{1-f} + \tfrac{n_3(f,\alpha,s)}{1+n_3(f,\alpha,s)}\right)
    \end{align*}
    where 
    \begin{align*}
        &n_3(f,\alpha,s) =\tfrac{\sqrt{(1+\alpha)(1-f)}-3 +\sqrt{\left(\sqrt{(1+\alpha)(1-f)}-1\right)^2 + \frac{4\sqrt{(1+\alpha)(1-f)}}{1-s}}}{2}.
    \end{align*}

    Through a combination, we obtain that the fees collected across both pools are given by
    \begin{align*}
    &F (f,\alpha,s,y,p,\omega)\\ 
    &= \begin{cases} 
      0 &\quad0< \alpha \leq \alpha^{\min}_N\\
      \parbox[t]{.33\linewidth}{$F_N+\omega\cdot F_{W,R}^{U^A_R\leq 0}$}   &\quad \alpha^{\min}_N<\alpha \leq\alpha^{\min}_W  \text{ \& } U^A_R\leq 0\\
      \parbox[t]{.33\linewidth}{$F_N+\omega\cdot F_{W,R}^{U^A_R> 0}$}   &\quad \alpha^{\min}_N<\alpha \leq\alpha^{\min}_W  \text{ \& } U^A_R> 0\\
      \parbox[t]{.33\linewidth}{$F_N+ (1-\omega)F_W^{U^A\leq 0}$\\$+ \omega\cdot F_{W,R}^{U^A_R\leq 0}$}  &\quad \alpha \geq \alpha^{\min}_W  \text{ \& }  U^A_R,U^A_S\leq 0\\
      \parbox[t]{.33\linewidth}{$F_N+ (1-\omega)F_W^{U^A\leq 0}$\\$+ \omega\cdot F_{W,R}^{U^A_R> 0}$}  &\quad \alpha \geq \alpha^{\min}_W  \text{, }  U^A_R>0\text{ \& }U^A_S\leq 0\\
      \parbox[t]{.33\linewidth}{$F_N+ (1-\omega)F_W^{U^A> 0}$\\$+ \omega\cdot F_{W,R}^{U^A_R> 0}$}  &\quad \alpha \geq \alpha^{\min}_W  \text{ \& }  U^A_R,U^A_S>0\\
    \end{cases}
    \end{align*}

    $F (f,\alpha,s,y,p,\omega)$ linearly combines the fees from the two streams (sophisticated and retail) according to their relative proportions (i.e., sophisticated trades make up a proportion $1-\omega$ and retail traders the rest. We conclude that $F$ is proportional to $p$ for every $\alpha>0$.
\end{proof}

Importantly, Lemma~\ref{lem:totalfee} holds an infinite sequence of trade orders from sophisticated and retail traders with the same ($\alpha$,$s$) along with the associated orders from sandwich attackers and arbitrageurs. The previous follows from price arbitrageurs returning the pool to its initial price $P_{X\rightarrow Y}$ after every trade sequence. We conclude that the total fees collected for a continuous stream of trade orders originating from a homogeneous set of traders with the same relative benefit $\alpha$ and slippage tolerance $s$ is proportional to $p$, i.e., the fraction of the total liquidity placed in $\text{Pool}_N$.

Lemma~\ref{lem:totalfee} gives the system's total fees $ F(f,\alpha,s,y,p,\omega)$. By virtue of the proportionality of the total fees $F$ in both $p$ and $y$, the fees received by an individual liquidity provider $LP_i$ with liquidity $(p_i l_i L,(1-p_i) l_i L)$ are given by $F_{i}(f,\alpha,s,y,p_i,l_i,\omega) = l_i \cdot F(f,\alpha,s,y,p_i,\omega)$. Therefore, the optimal liquidity distribution for an individual liquidity provider also has all liquidity in $\text{Pool}_N$ ($p_i=1$) or all liquidity in $\text{Pool}_W$ ($p_i=0$), whenever the gradient of the fee with respect to $p_i$ is non-zero. A liquidity provider will redistribute their liquidity optimally, i.e., such that their utility is maximized, whenever they can increase their received fees by more than a factor of $1+\varepsilon$.

%-------------------------------------------------------------------------------
\section{Game Equilibria}\label{sec:equilibria}
%-------------------------------------------------------------------------------
Before discussing the quantitative model, we will give an intuitive explanation. First, we note that liquidity providers profit from large trading volumes, irrespective of their origin. As the sandwich attackers extract their profits from the traders, sophisticated traders will reduce their trading volume if the attacks become too lucrative but retail trades will not. Therefore, whether a pool with sandwich attacks is the equilibrium boils down to whether the increased trading volume the attackers generate can offset the diminished trading activity of sophisticated traders.
 
We will now substantiate this qualitative picture by locating the game's $\varepsilon$-equilibria to identify which pool is favored by the market actors. A liquidity distribution is an $\varepsilon$-equilibria if no liquidity provider can increase their utility by more than a factor of $1+\varepsilon$ by adjusting the liquidity distribution. For $\varepsilon=0$, any $\varepsilon$-equilibrium is considered a Nash equilibrium. We analyze the game's equilibria assuming a fixed (mean) benefit for the traders, $\alpha$, in Section~\ref{sec:homogenous} and discuss the heterogeneous case in Section~\ref{sec:HeterogeneousTraders}. Further note that the section's omitted proofs can be found in Appendix~\ref{app:game}.
\subsection{Homogeneous Traders}\label{sec:homogenous}
We start by analyzing the game's equilibria given a homogeneous trader set, i.e., all traders have the same relative benefit $\alpha$. In the simplest case, $\partial _p F =0$, all liquidity distributions are both $\varepsilon$-equilibria and Nash equilibria (cf. Lemma~\ref{lem:strongNash}). 
\vspace{-3pt}

\begin{restatable}[]{lemma}{strongNash}\label{lem:strongNash}
    The only Nash equilibria if $\partial _p F \neq 0$ are $p \in \{0,1\}$. If $\partial _p F =0$, all liquidity distributions are $\varepsilon$-equilibria in a homogeneous traders game.
\end{restatable}

\begin{figure}[t]
\centering
  \begin{subfigure}[t]{1\linewidth}
  \includegraphics[scale=1]{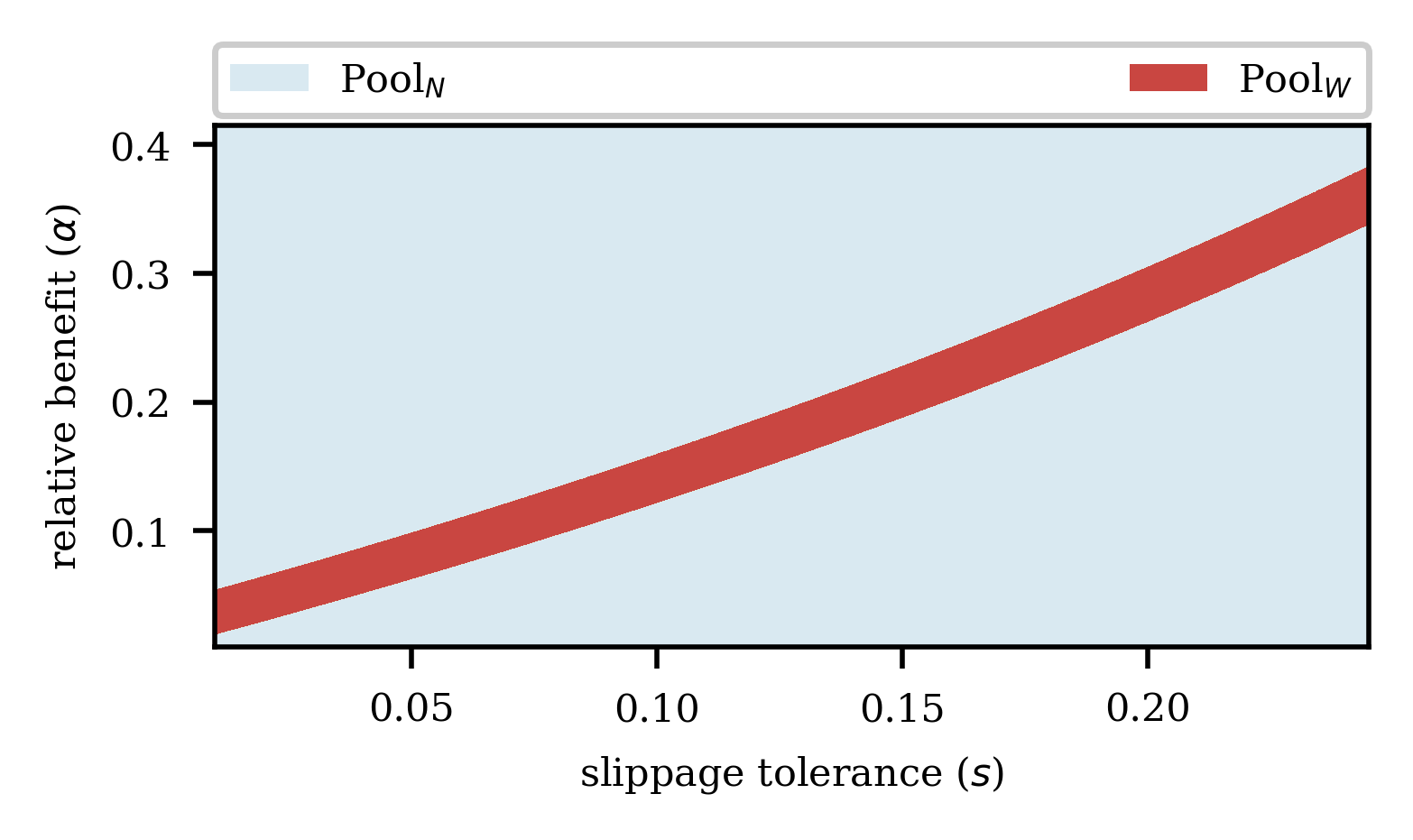}\vspace{-6pt}
  
    \caption{We set $\omega= 0.01$.} \label{fig:gradFeeTrade1}
  \end{subfigure}

  \begin{subfigure}[t]{1\linewidth}
    \includegraphics[scale=1]{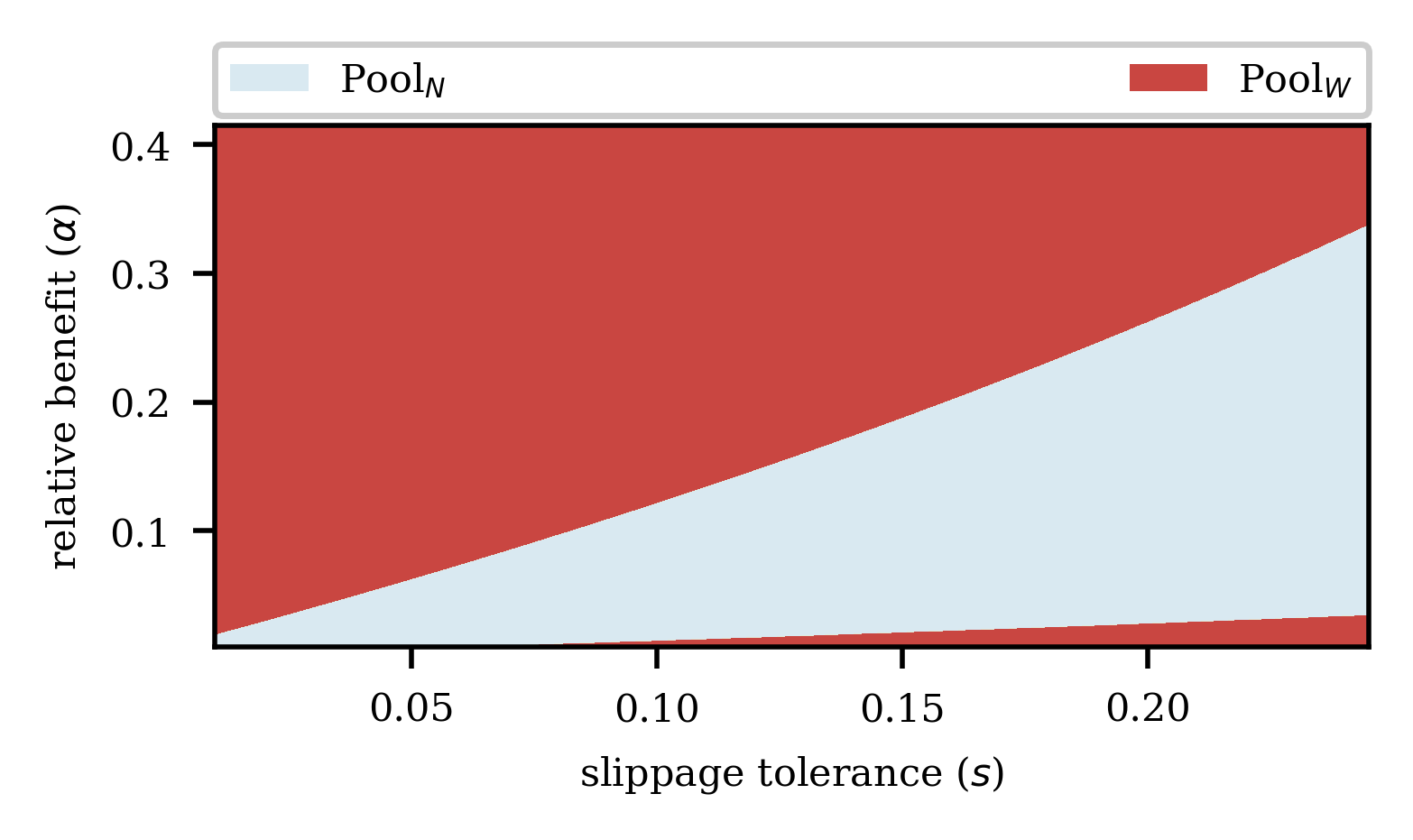}\vspace{-6pt}
    
    \caption{We set $\omega= 0.1$.} \label{fig:gradFeeTrade10}
  \end{subfigure}\vspace{-4pt}
  \caption{The Nash equilibrium (color shading), is dependent on the slippage tolerance and the relative benefit for $\omega=0.01$ (cf. Figure~\ref{fig:gradFeeTrade1}) and $\omega=0.1$ (cf. Figure~\ref{fig:gradFeeTrade10}). We set $x = 5,000,000$ $X$, $y = 5,000,000$ $Y$ and $f= 0.003$.} \label{fig:NashEq}\vspace{-6pt}
\end{figure}

In Lemma~\ref{lem:strongNash} we further show that for $\partial _p F \neq 0$ the only possible Nash equilibria are the two corner cases: all liquidity in $\text{Pool}_N$ ($p=1$) or in $\text{Pool}_W$ ($p=0$). This follows from the proportionality of the fees to $p$ which means that the sign of $\partial _p F$ dictates the location of the Nash equilibrium. In Figure~\ref{fig:NashEq}, we plot the dependence of this equilibrium on the slippage tolerance and relative benefit for $\omega=0.01$ (i.e., 1\% of the trade flow originates from retail traders) and $\omega=0.1$ (i.e., 10\% of the trade flow originates from retail traders).

We first consider the setting when $\omega=0.01$ (cf. Figure~\ref{fig:gradFeeTrade1}). Notice that in areas where either the trader's relative benefit $\alpha$ or the slippage tolerance $s$ is high, $\text{Pool}_N$ is the Nash equilibrium. When $\alpha$ is comparatively large, so is the traders' transaction size. Liquidity providers, therefore, receive a substantial amount of fees from sophisticated traders and sandwich attacks would decrease the pool's trading volume more than the volume created by the attacker. Thus, all liquidity is in $\text{Pool}_N$. The same holds when the slippage tolerance is high compared to the trader's benefit. Sophisticated traders no longer trade in $\text{Pool}_W$ (cf. Lemma~\ref{lem:tradesizesoph}) or their size in $\text{Pool}_W$ is too small for there to be a profitable sandwich attack (cf. Lemma~\ref{lem:smalldelta}). Thus, the only volume in $\text{Pool}_W$ stems from retail traders (cf. Lemma~\ref{lem:tradesizeretail}). 

There is a small area in between where $\text{Pool}_W$ is the equilibrium. Here, the slippage tolerance is just small enough not to exceed the bound given in Lemma~\ref{lem:tradesizesoph} and the sophisticated trader's transaction size in $\text{Pool}_W$ is just large enough to allow for a profitable attack (cf. Lemma~\ref{lem:smalldelta}). Further, the trades from retail traders also allow for profitable attacks as their trade size exceeds that of sophisticated traders. Thus, liquidity providers' private incentives are maximized in the presence of sandwich attackers. 

When $\omega=0.1$ the picture shifts and $\text{Pool}_W$ is the Nash equilibrium in a larger proportion for the parameter space due to the increase in volume from retail traders  (cf. Figure~\ref{fig:gradFeeTrade10}. For relatively large $\alpha$, $\text{Pool}_W$ dominates as the Nash equilibrium. In this part of our parameter space, the trades from retail and sophisticated traders suffer from sandwich attacks. Thus, for relatively large $\alpha$ the extra volume associated with retail traders, in comparison to the setting in Figure~\ref{fig:gradFeeTrade1}, prevents $\text{Pool}_N$ from being the Nash equilibrium.

In the part of the parameter space where the slippage tolerance is large in comparison to the relative benefit, we observe that in a sliver of the space $\text{Pool}_W$ is the Nash equilibrium. Here, the sandwich attack on the trades from retail traders are profitable. Further, the sandwich attack volume is relatively large in comparison to the volume from sophisticated and retail traders in $\text{Pool}_N$ as the slippage tolerance is large in comparison to the relative benefit. However, when this difference decreases, $\text{Pool}_N$ becomes the Nash equilibrium as the sandwich attack volume can no longer compensate the loss in volume from sophisticated traders.

As $\omega$ grows the parts of the parameter space where $\text{Pool}_W$ is the Nash equilibrium move in on each other. In particular, for $\omega=1$ (i.e., there are only retail traders) $\text{Pool}_W$ is the Nash equilibrium on the entire parameter space and the opposite holds for $\omega=0$ (i.e., there are no retail traders). Thus, the composition of order flow from the two types of traders is a major factor determining where the Nash equilibrium lies. 

While the sign of the gradient $\partial _p F$ dictates the location of the Nash equilibrium, it is not sufficient to determine if it is an $\varepsilon$-equilibrium. As we show in Theorem~\ref{thm:softNash}, it is the relative difference between the fees the liquidity provider earns with their current distribution and the maximum fees they can collect that dictate whether the liquidity provider will change strategy. 
\begin{restatable}[]{theorem}{softNash}\label{thm:softNash}
    A liquidity distribution is an $\varepsilon$-equilibrium if there is no liquidity provider with initial liquidity distribution $(p_i l_i L,(1-p_i) l_i L)$, such that
    $$\frac{\max\{F (f,\alpha,s,y,0,\omega) ,F (f,\alpha,s,y,1,\omega)  \} }{  F \left(f,\alpha,s,y,p_i,\omega\right)}-1<  \varepsilon.$$
\end{restatable}

Currently, all liquidity is in markets that allow for sandwich attacks. Therefore, even if the Nash equilibrium is in $\text{Pool}_N$, liquidity providers would have to move their liquidity.

Thus, we also identify market configurations that are $\varepsilon$-equilibria independent of the current liquidity distribution. In this situation, a new market with a front-running protection mechanism would not attract any liquidity even if it were to maximize the liquidity provider's private incentives. The maximum relative change in fees for a  given market configuration is given as $\vert\Delta_F\vert$, where
$\Delta_F = {\partial _p F}/ F_{\min}$ and $ F_{\min}={\min\left(F(p=0),F(p=1)\right)}.$
Independent of a liquidity provider's initial distribution, the relative benefit of switching strategy cannot exceed $\varepsilon$ in case $\vert \Delta_ F\vert <\varepsilon$. The sign of $\Delta_F$ corresponds to the sign of the fee's gradient $\partial _p F$ and therefore indicates the position of the Nash equilibrium. 

We simulate the dependence of $\Delta_F$ on the slippage tolerance and relative benefit in Figure~\ref{fig:feeGradient} for $\omega=0.01$ and $\omega=0.1$. Starting with the setting where $\omega=0.01$ (cf. Figure~\ref{fig:signFee1}). For comparatively large slippage tolerances, the magnitude of $\Delta_F$ is large. Independent of the liquidity in $\text{Pool}_W$, the trading volume from sophisticated traders is either zero when $\alpha <\alpha_W^{\min}$ or relatively small when sandwich attacks are not profitable. Therefore, switching strategies by moving liquidity from $\text{Pool}_W$ to $\text{Pool}_N$ leads to a sizable increase in fees in this part of the parameter space, and we do not expect any $\varepsilon$-equilibria in $\text{Pool}_W$ for this parameter range. 

\begin{figure}[t]
\centering
      \begin{subfigure}[t]{1\linewidth}
  \includegraphics[scale=1]{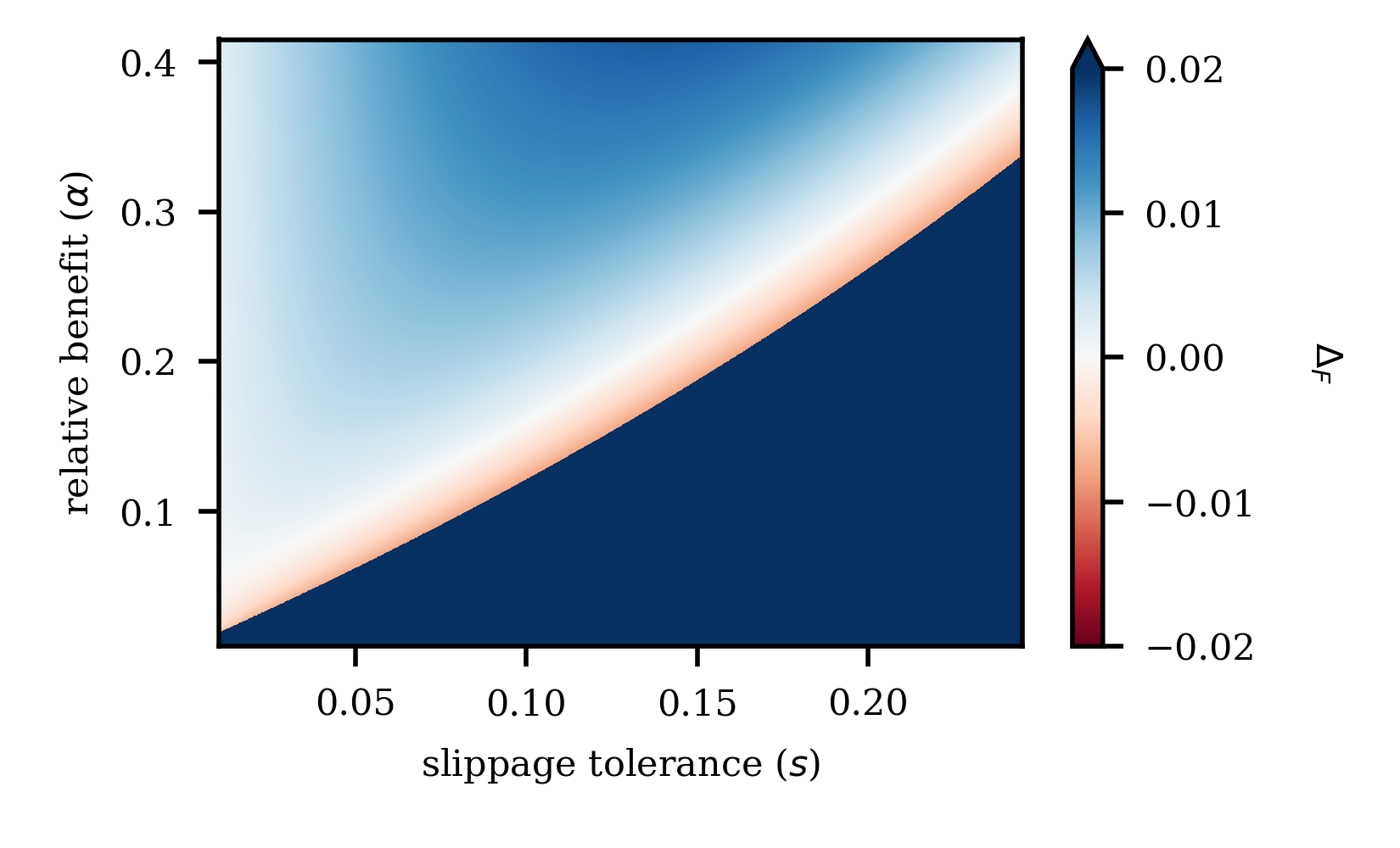}\vspace{-6pt}
  
    \caption{We set $\omega= 0.01$.} \label{fig:signFee1}
  \end{subfigure}

  \begin{subfigure}[t]{1\linewidth}
    \includegraphics[scale=1]{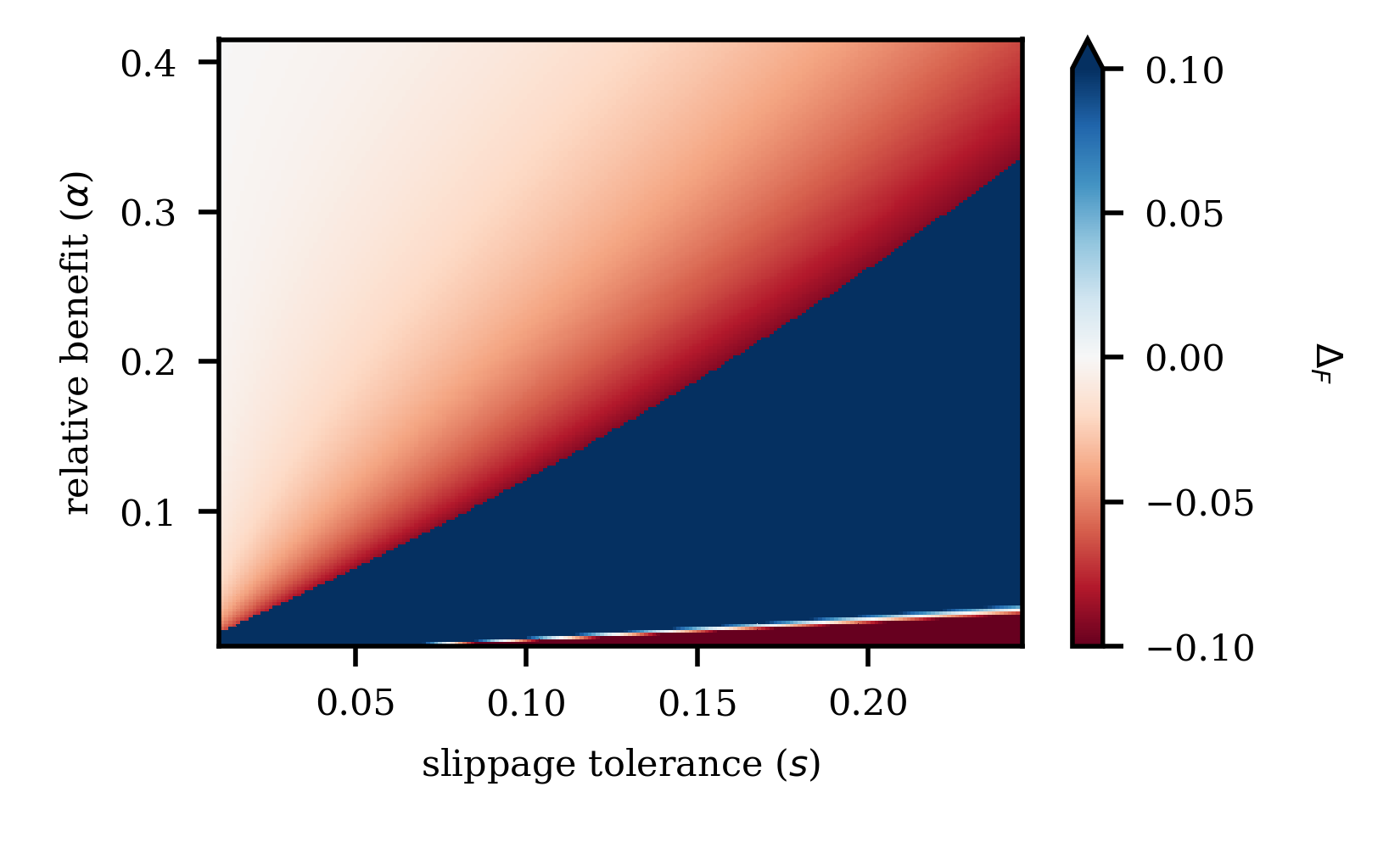}\vspace{-6pt}
    
    \caption{We set $\omega= 0.1$.} \label{fig:signFee10}
  \end{subfigure}\vspace{-4pt}
\caption{Simulation of $ \Delta_ F$ across both pools depending on the trader's relative benefit and the slippage tolerance for $\omega=0.01$ (cf. Figure~\ref{fig:signFee1}) and $\omega=0.1$ (cf. Figure~\ref{fig:signFee10}). In blue areas, the Nash equilibrium is $\text{Pool}_N$, in red areas, it is $\text{Pool}_W$. $ \Delta_ F$ is cut off for better visibility and notice that the cutoff is different in the two plots. We set $x = 5,000,000$ $X$, $y = 5,000,000$ $Y$ and $f= 0.003$.} \label{fig:feeGradient}\vspace{-10pt}
\end{figure}

Turning to more realistic areas of the parameter space where the relative benefit is larger than the slippage tolerance, we notice that  $\Delta_F$'s magnitude is small (i.e., $\Delta_F < 0.02$). Thus, all liquidity providers who only change strategies for a relative benefit larger than 2\% would not be inclined to move their liquidity. We follow that any liquidity distribution is an $\varepsilon$-equilibrium for a significant proportion of the parameter space even for small $\varepsilon$. 

For $\omega=0.1$ we observe a different picture (cf. Figure~\ref{fig:signFee10}). Recall, that in general for a higher $\omega$ a larger proportion of the parameter space has $\text{Pool}_W$ (red areas) as the Nash equilibrium. In addition to that, we observe that in general  $\Delta_F$'s magnitude is larger. Only when the relative benefit is very large in comparison to the slippage tolerance is the relative benefit small, i.e., below  2\%. 

Therefore, when $\omega=0.01$ liquidity providers are largely indifferent to whether the market utilizes a front-running protection mechanism, and might require additional financial incentives to migrate their liquidity to pools with front-running protection mechanisms. However, when $\omega=0.1$ the difference between the pools are more pronounced and liquidity providers are more likely to move their liquidity to the pool with higher fee revenue. This, however, is $\text{Pool}_W$ in the most realistic areas of the parameter space.  

\subsection{Heterogeneous Traders}\label{sec:HeterogeneousTraders}

We continue with analyzing where the $\varepsilon$-equilibria fall in a game with heterogeneous traders. We model a trader's relative benefit $\alpha$ as a random variable $A$ with probability mass function $\psi_A(\alpha)$. The game is in an $\varepsilon$-equilibrium for any probability mass function $\psi_A(\alpha)$ that fulfills the condition provided in Theorem~\ref{thm:epsNashdist}. Theorem~\ref{thm:epsNashdist} assumes that the random variable $A$ is discrete. Note, however, that it could be adapted to the continuous case. 
\begin{restatable}[]{theorem}{epsNashdist}\label{thm:epsNashdist}
    A liquidity distribution in a system with heterogeneous traders with distribution $\psi_A(\alpha)$ is an $\varepsilon$-equilibrium if there is no liquidity provider with initial liquidity distribution $(p_i l_i L,(1-p_i) l_i L)$, such that
    
    \begin{gather*}
        \frac {\max\left\{\sum _\alpha \psi_A(\alpha)F (f,\alpha,s,y,0,\omega)  ,\sum _\alpha\psi_A(\alpha)F (f,\alpha,s,y,1,\omega)  \right\}}{ \sum _\alpha \psi_A(\alpha)F \left(f,\alpha,s,y,p_i,\omega\right)  } \\- 1<  \varepsilon.
    \end{gather*}
\end{restatable}

Further, we predict that for most probability distributions, extreme values of the system's fee gradient $\partial _p F$ will be averaged out. To back up this assumption, we simulate the $\Delta_F$ for a two-point distribution in Appendix~\ref{app:hetero}.

%-------------------------------------------------------------------------------
\section{Social Incentives and Self-Regulation}
%-------------------------------------------------------------------------------

In our model, whenever the derivative of the liquidity provider's utility is positive, i.e., $\partial _p U^{LP}=\partial _pF >0$, the system's Nash equilibrium maximizes social welfare. 

Our analysis demonstrates that without or with very few retail traders, i.e., traders that act irrational, $\text{Pool}_N$ is the Nash equilibrium for the vast majority of the parameter space (cf. Figure~\ref{fig:gradFeeTrade1}). Therefore, markets preventing front-running generally align the private incentives of liquidity providers with the system's incentives. However, liquidity providers are currently in markets without front-running protections. Thus, an innovative DEX preventing front-running attacks must attract liquidity from other markets. Our analysis highlights that even when placing all liquidity in $\text{Pool}_N$ maximizes the private incentives of liquidity providers, the benefit from adjusting a liquidity distribution is often only small. The market, therefore, cannot rely on the inert liquidity providers to revise their liquidity distribution. Therefore, added incentives may be required for the successful adoption of a such market. 

Additionally, when there is a significant proportion of retail traders $\text{Pool}_W$ is the Nash equilibrium more frequently (cf. Figure~\ref{fig:gradFeeTrade1}). This is a consequence of ``irrational'' behavior from retail traders who do not respond to the presence of sandwich attacks -- possibly due to an information asymmetry. 

Thus, for the market to self-regulate it must (1) attract liquidity to novel DEXes, and (2) educate traders. 

\Paragraph{Attracting Liquidity Providers} In the following, we discuss the possibilities of how a new DEX, implementing a front-running prevention scheme, could attract liquidity providers. One possibility would be for the DEX to distribute its native token to liquidity providers as an added incentive. Similar benefits have been distributed at the launch of new DeFi platforms (cf. Sushiswap~\cite{2021sushiswap}). Note here that it would be important for these distributed tokens to cover at least the gas fees required to migrate the liquidity, as otherwise, it is unlikely that it would motivate many. Another possibility would be to directly cover the migration costs. Such offers are utilized by brokerages in traditional finance (cf. Ally~\cite{2022ally}) and could also be implemented by DEXes. Finally, once they reach a certain traction we would expect many liquidity providers to follow. 

\Paragraph{Educating Traders} In our model, the trade volume retail traders direct to the respective pools is proportional to the liquidity available in that pool, thus they completely ignore the effects of sandwich attacks. Further, this behavior is a major driver in the Nash equilibria lying in pools with sandwich attacks. The behavior of retail traders, who are likely unaware of the presence and consequences of the attacks, can be altered through education. Retail traders adjusting their response to sandwich attacks is crucial for the adoption of a prevention mechanism. 

%-------------------------------------------------------------------------------
\section{Conclusion}
%-------------------------------------------------------------------------------
Our game-theoretical study of the incentives of traders and liquidity providers to adopt a DEX with a new market design preventing front-running attacks shows that when the vast majority of traders ($\approx$99\%) are sophisticated, the private incentives of both traders and liquidity providers generally align the market's social incentives --- eliminating front-running attacks. However, this drastically shifts when the proportion of retail traders increases. Even when retail traders only account for 10\% of the order flow, the private incentives of liquidity providers oppose the market's social incentives, i.e., liquidity providers are generally drawn to pools with sandwich attacks. 

This finding highlights the struggles of eliminating front-running attacks. Even though, the private incentives of sophisticated traders and liquidity providers align (i.e., without retail traders the Nash equilibrium is always in pools without sandwich attacks), a small proportion of traders who act irrationally, completely changes the picture. 

Further, even if the Nash equilibrium is in pools without sandwich attacks there is a further challenge. In the absence of a central authority, market participants must experience a personal benefit for the successful adoption of such a design. The alignment of the liquidity provider's private incentives and the market's social incentives is promising. Yet, our analysis also finds that the increase in the liquidity provider's utility from moving to a market preventing front-running is generally small. Liquidity providers are usually not nimble market participants. Therefore, the prospect of a small utility increase might not suffice.

Successful self-regulation of the market to prevent front-running attacks is likely to not only require intensive education of traders but also additional initial financial incentives to gain the attention of liquidity providers. 

%\section{Ethics}
%Given the theoretical nature and topic of this work, it does not raise any ethical issues.

%\section{Open Science}
%We will have scripts used for the simulations evaluated and will make them openly available. 
%-------------------------------------------------------------------------------
\bibliographystyle{plain}
\bibliography{bibliography}

\begin{thebibliography}{10}

\bibitem{2022ally}
Ally.
\newblock \url{https://www.ally.com/help/invest/transfers/}, 2024.

\bibitem{2022ata}
Automata network.
\newblock \url{https://www.ata.network/}, 2024.

\bibitem{2022bloXroute}
bloxroute.
\newblock \url{https://bloxroute.com/private-transactions/}, 2024.

\bibitem{2022cowswap}
Cowswap.
\newblock \url{https://cowswap.exchange/}, 2024.

\bibitem{2021defillama}
Dexs volume.
\newblock \url{https://defillama.com/dexs}, 2024.

\bibitem{2021eden}
Eden.
\newblock \url{https://www.edennetwork.io/}, 2024.

\bibitem{2021flash}
flashbots.
\newblock \url{https://docs.flashbots.net/}, 2024.

\bibitem{2022gnosis}
Gnosis protocol.
\newblock \url{https://gnosis.io/}, 2024.

\bibitem{2022openmev}
Openmev.
\newblock \url{https://openmev.xyz/}, 2024.

\bibitem{PBSEthereum2023}
Proposer-builder separation, 2024.

\bibitem{2022eigenphi}
Sandwich overview.
\newblock \url{https://eigenphi.io/mev/ethereum/sandwich}, 2024.

\bibitem{2022secretswap}
Secretswap.
\newblock \url{https://secretswap.net/}, 2024.

\bibitem{2021sushiswap}
Sushiswap.
\newblock \url{https://sushi.com/}, 2024.

\bibitem{adams2020uniswap}
Hayden Adams, Noah Zinsmeister, and Dan Robinson.
\newblock Uniswap v2 core.
\newblock 2020.

\bibitem{adams2021uniswap}
Hayden Adams, Noah Zinsmeister, Moody Salem, River Keefer, and Dan Robinson.
\newblock Uniswap v3 core.
\newblock Technical report, Uniswap, 2021.

\bibitem{asayag2018fair}
Avi Asayag, Gad Cohen, Ido Grayevsky, Maya Leshkowitz, Ori Rottenstreich, Ronen Tamari, and David Yakira.
\newblock A fair consensus protocol for transaction ordering.
\newblock In {\em 2018 IEEE 26th International Conference on Network Protocols (ICNP)}, pages 55--65, 2018.

\bibitem{baird2016swirlds}
Leemon Baird.
\newblock The swirlds hashgraph consensus algorithm: Fair, fast, byzantine fault tolerance.
\newblock {\em Swirlds Tech Reports SWIRLDS-TR-2016-01, Tech. Rep}, 2016.

\bibitem{bentov2019tesseract}
Iddo Bentov, Yan Ji, Fan Zhang, Lorenz Breidenbach, Philip Daian, and Ari Juels.
\newblock Tesseract: Real-time cryptocurrency exchange using trusted hardware.
\newblock In {\em Proceedings of the 2019 ACM SIGSAC Conference on Computer and Communications Security}, CCS '19, page 1521–1538, New York, NY, USA, 2019. Association for Computing Machinery.

\bibitem{BERNHARDT2008front}
Dan Bernhardt and Bart Taub.
\newblock Front-running dynamics.
\newblock {\em Journal of Economic Theory}, 138(1):288--296, 2008.

\bibitem{breidenbach2018enter}
Lorenz Breidenbach, Phil Daian, Florian Tram{\`e}r, and Ari Juels.
\newblock Enter the hydra: Towards principled bug bounties and exploit-resistant smart contracts.
\newblock In {\em 27th $\{$USENIX$\}$ Security Symposium ($\{$USENIX$\}$ Security 18)}, pages 1335--1352, 2018.

\bibitem{budisch2019theory}
Eric Budish, Robin~S. Lee, and John~J. Shim.
\newblock {A Theory of Stock Exchange Competition and Innovation: Will the Market Fix the Market?}
\newblock NBER Working Papers 25855, National Bureau of Economic Research, Inc, 2019.

\bibitem{cachin2021quick}
Christian Cachin, Jovana Mi{\'c}i{\'c}, and Nathalie Steinhauer.
\newblock Quick order fairness.
\newblock In {\em Financial Cryptography and Data Security (FC), Grenada}, 2022.

\bibitem{comerton2007anonymity}
Carole Comerton-Forde and Kar~Mei Tang.
\newblock Anonymity, frontrunning and market integrity.
\newblock {\em The Journal of Trading}, 2(4):101--118, 2007.

\bibitem{constantinescu2023fair}
Andrei Constantinescu, Diana Ghinea, Lioba Heimbach, Zilin Wang, and Roger Wattenhofer.
\newblock A fair and resilient decentralized clock network for transaction ordering.
\newblock In {\em 27th International Conference on Principles of Distributed Systems (OPODIS), Tokyo, Japan}, December 2023.

\bibitem{daian2020flash}
Philip Daian, Steven Goldfeder, Tyler Kell, Yunqi Li, Xueyuan Zhao, Iddo Bentov, Lorenz Breidenbach, and Ari Juels.
\newblock Flash boys 2.0: Frontrunning in decentralized exchanges, miner extractable value, and consensus instability.
\newblock In {\em 2020 IEEE Symposium on Security and Privacy (SP)}, pages 910--927. IEEE, 2020.

\bibitem{danthine1998front}
Jean-Pierre Danthine and Serge Moresi.
\newblock {Front-Running by Mutual Fund Managers: A Mixed Bag}.
\newblock {\em Review of Finance}, 2(1):29--56, 1998.

\bibitem{doweck2020multi}
Yael Doweck and Ittay Eyal.
\newblock Multi-party timed commitments.
\newblock {\em arXiv preprint arXiv:2005.04883}, 2020.

\bibitem{eskandari2019sok}
Shayan Eskandari, Mahsa Moosavi, and Jeremy Clark.
\newblock Sok: Transparent dishonesty: front-running attacks on blockchain.
\newblock In {\em Financial Cryptography and Data Security (FC), St. Kitts, Saint Kitts and Nevis}, February 2019.

\bibitem{heimbach2022eliminating}
Lioba Heimbach and Roger Wattenhofer.
\newblock Eliminating sandwich attacks with the help of game theory.
\newblock In {\em ACM Asia Conference on Computer and Communications Security (ASIA CCS), Nagasaki, Japan}, June 2022.

\bibitem{heimbach2022sok}
Lioba Heimbach and Roger Wattenhofer.
\newblock {SoK: Preventing Transaction Reordering Manipulations in Decentralized Finance}.
\newblock In {\em {4th ACM Conference on Advances in Financial Technologies (AFT), Cambridge, Massachusetts, USA}}, September 2022.

\bibitem{kelkar2021order}
Mahimna Kelkar, Soubhik Deb, and Sreeram Kannan.
\newblock Order-fair consensus in the permissionless setting.
\newblock {\em IACR Cryptol. ePrint Arch.}, 2021:139, 2021.

\bibitem{kelkar2021themis}
Mahimna Kelkar, Soubhik Deb, Sishan Long, Ari Juels, and Sreeram Kannan.
\newblock Themis: Fast, strong order-fairness in byzantine consensus.
\newblock Cryptology ePrint Archive, Report 2021/1465, 2021.
\newblock \url{https://ia.cr/2021/1465}.

\bibitem{kelkar2020order}
Mahimna Kelkar, Fan Zhang, Steven Goldfeder, and Ari Juels.
\newblock Order-fairness for byzantine consensus.
\newblock In {\em Annual International Cryptology Conference}, pages 451--480. Springer, 2020.

\bibitem{kursawe2020wendy}
Klaus Kursawe.
\newblock Wendy, the good little fairness widget: Achieving order fairness for blockchains.
\newblock In {\em Proceedings of the 2nd ACM Conference on Advances in Financial Technologies}, pages 25--36, 2020.

\bibitem{manahov2016front}
Viktor Manahov.
\newblock Front-running scalping strategies and market manipulation: Why does high-frequency trading need stricter regulation?
\newblock {\em Financial Review}, 51(3):363--402, 2016.

\bibitem{markham1988front}
Jerry~W. Markham.
\newblock Front-running - insider trading under the commodity exchange act.
\newblock {\em Catholic University Law Review}, 38:69, 1988-1989.

\bibitem{miller2016honey}
Andrew Miller, Yu~Xia, Kyle Croman, Elaine Shi, and Dawn Song.
\newblock The honey badger of bft protocols.
\newblock CCS '16, page 31–42, New York, NY, USA, 2016. Association for Computing Machinery.

\bibitem{momeni2023fairblock}
Peyman Momeni, Sergey Gorbunov, and Bohan Zhang.
\newblock Fairblock: Preventing blockchain front-running with minimal overheads.
\newblock In {\em Security and Privacy in Communication Networks: 18th EAI International Conference, SecureComm 2022, Virtual Event, October 2022, Proceedings}, pages 250--271. Springer, 2023.

\bibitem{moosa2015regulation}
Imad Moosa.
\newblock The regulation of high-frequency trading: A pragmatic view.
\newblock {\em Journal of Banking Regulation}, 16(1):72--88, 2015.

\bibitem{orda2021enforcing}
Ariel Orda and Ori Rottenstreich.
\newblock Enforcing fairness in blockchain transaction ordering.
\newblock {\em Peer-to-peer Networking and Applications}, 14(6):3660--3673, 2021.

\bibitem{park2021conceptual}
Andreas Park.
\newblock The conceptual flaws of constant product automated market making, 2021.
\newblock Available at SSRN: 3805750.

\bibitem{qin2021quantifying}
Kaihua Qin, Liyi Zhou, and Arthur Gervais.
\newblock Quantifying blockchain extractable value: How dark is the forest?
\newblock In {\em 2022 IEEE Symposium on Security and Privacy (SP)}, pages 198--214. IEEE, 2022.

\bibitem{reiter1994securely}
Michael~K Reiter and Kenneth~P Birman.
\newblock How to securely replicate services.
\newblock {\em ACM Transactions on Programming Languages and Systems (TOPLAS)}, 16(3):986--1009, 1994.

\bibitem{stathakopoulou2021adding}
Chrysoula Stathakopoulou, Signe R{\"u}sch, Marcus Brandenburger, and Marko Vukoli{\'c}.
\newblock Adding fairness to order: Preventing front-running attacks in bft protocols using tees.
\newblock In {\em 2021 40th International Symposium on Reliable Distributed Systems (SRDS)}, pages 34--45. IEEE, 2021.

\bibitem{tatabitovska2021mitigation}
Ana Tatabitovska, Oğuzhan Ersoy, and Zekiraya Erkin.
\newblock Mitigation of transaction manipulation attacks in uniswap.
\newblock 2021.

\bibitem{Wang2022impact}
Ye~Wang, Patrick Züst, Yaxing Yao, Zhicong Lu, and Roger Wattenhofer.
\newblock {Impact and User Perception of Sandwich Attacks in the DeFi Ecosystem}.
\newblock In {\em {ACM CHI Conference on Human Factors in Computing Systems (CHI), New Orleans, LA, USA}}, May 2022.

\bibitem{wood2014ethereum}
Gavin Wood et~al.
\newblock Ethereum: A secure decentralised generalised transaction ledger.
\newblock 2014.

\bibitem{zhang2022flash}
Haoqian Zhang, Louis-Henri Merino, Vero Estrada-Galinanes, and Bryan Ford.
\newblock Flash freezing flash boys: Countering blockchain front-running.
\newblock In {\em 2022 IEEE 42nd International Conference on Distributed Computing Systems Workshops (ICDCSW)}, pages 90--95. IEEE, 2022.

\bibitem{zhang2020byzantine}
Yunhao Zhang, Srinath Setty, Qi~Chen, Lidong Zhou, and Lorenzo Alvisi.
\newblock Byzantine ordered consensus without byzantine oligarchy.
\newblock In {\em 14th $\{$USENIX$\}$ Symposium on Operating Systems Design and Implementation ($\{$OSDI$\}$ 20)}, pages 633--649, 2020.

\bibitem{zhou2021a2mm}
Liyi Zhou, Kaihua Qin, and Arthur Gervais.
\newblock A2mm: Mitigating frontrunning, transaction reordering and consensus instability in decentralized exchanges, 2021.

\bibitem{zhou2020highfrequency}
Liyi Zhou, Kaihua Qin, Christof~Ferreira Torres, Duc~V Le, and Arthur Gervais.
\newblock High-frequency trading on decentralized on-chain exchanges, 2020.

\end{thebibliography}
\clearpage
\appendix

\section{Game Equilibria -- Heterogeneous Traders}\label{app:hetero}

We simulate the $\Delta_F$ for a two-point distribution with the following probability mass function:  
$$\psi_A ( \alpha) = \begin{cases}\frac{1}{2}&{\text{if }}\alpha=\alpha_k^-=\left( 1-\frac{1}{k}\right)\mu_\alpha ,\\\frac{1}{2}&{\text{if }}\alpha=\alpha_k^+ =\left( 1+\frac{1}{k}\right)\mu_\alpha,\end{cases} $$
in Figure~\ref{fig:signFeeDiffC} for $k=10$ (cf. Figure~\ref{fig:signFeeDiffc10}) and $k=3$ (cf. Figure~\ref{fig:signFeeDiffc3}). Note that for these simulations we set $\omega=0.01$, so retail traders account for only 1\% of the order flow. As expected $\Delta_F$ resembles the homogeneous case more closely, when the two points of the distribution are close to each other, i.e., for the higher values of $k$. Note that for $k = \infty$ the two-point distribution becomes a one-point distribution, i.e., the homogeneous case.

We further notice that a more significant area of the parameter space has $\text{Pool}_N$ as the Nash equilibrium when the slippage tolerance is large. Half the traders have a lower relative benefit than $\mu_\alpha$ and, thereby, these sophisticated traders will only execute transactions in $\text{Pool}_W$ for smaller slippage tolerances. We further note that the area of the parameter space, where $\text{Pool}_W$ is the Nash equilibrium grows smaller as $k$ decreases. While there remains a small area of the parameter space that has $\text{Pool}_W$ as the Nash equilibrium for $k=10$ (cf. Figure~\ref{fig:signFeeDiffc10}), this is noticeable smaller for $k=3$ (cf. Figure~\ref{fig:signFeeDiffc3}). 

The particular combination of requirements that must be met for $\text{Pool}_W$ to be the Nash equilibrium when $\omega$ is very small is achieved less frequently as the distance between the two points of the distribution grows. 

\begin{figure}[h!]
\centering
  \begin{subfigure}[t]{1\linewidth}
  \includegraphics[scale=1]{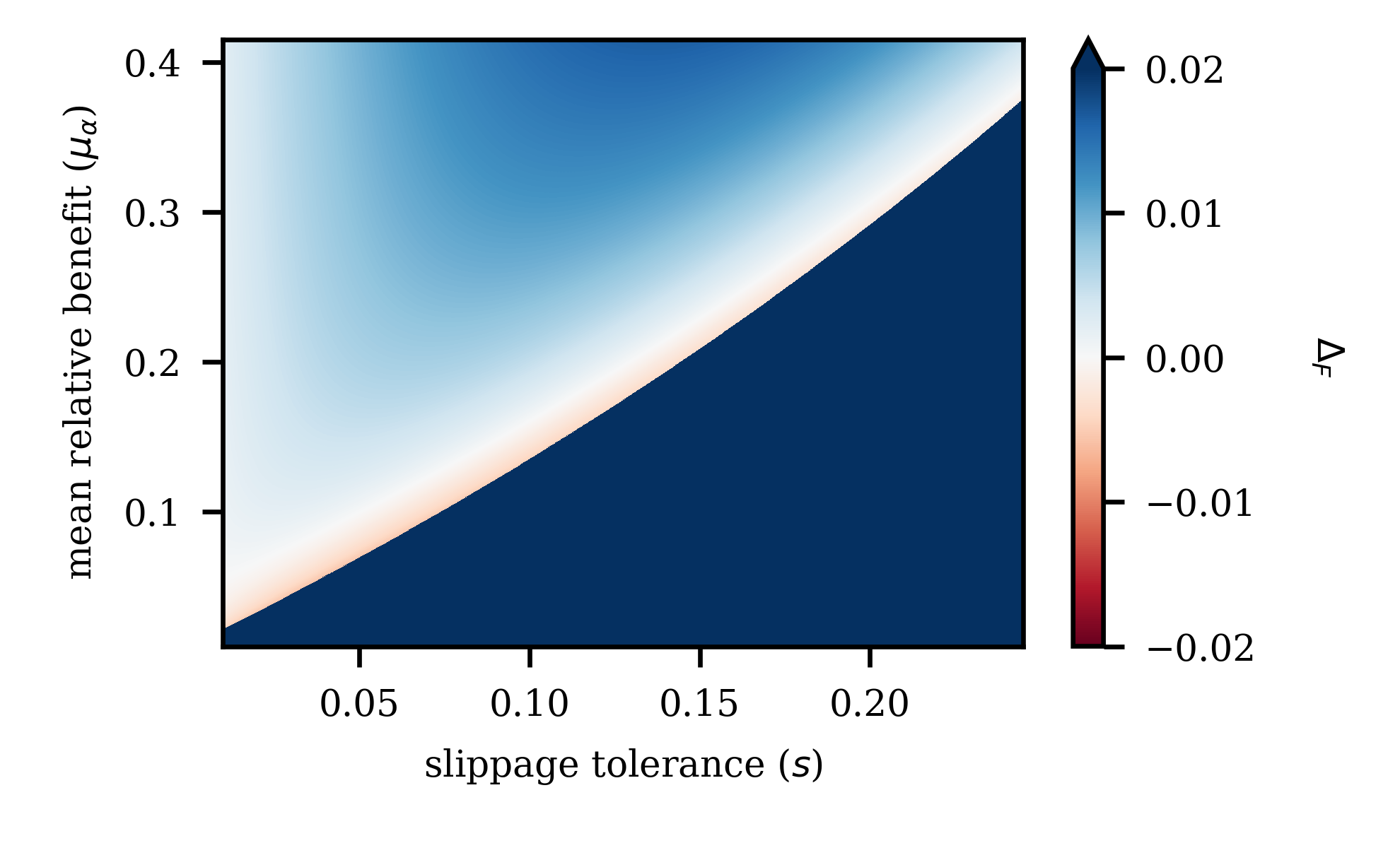}\vspace{-6pt}
  
    \caption{We set $k= 10$.} \label{fig:signFeeDiffc10}
  \end{subfigure}

  \begin{subfigure}[t]{1\linewidth}
    \includegraphics[scale=1]{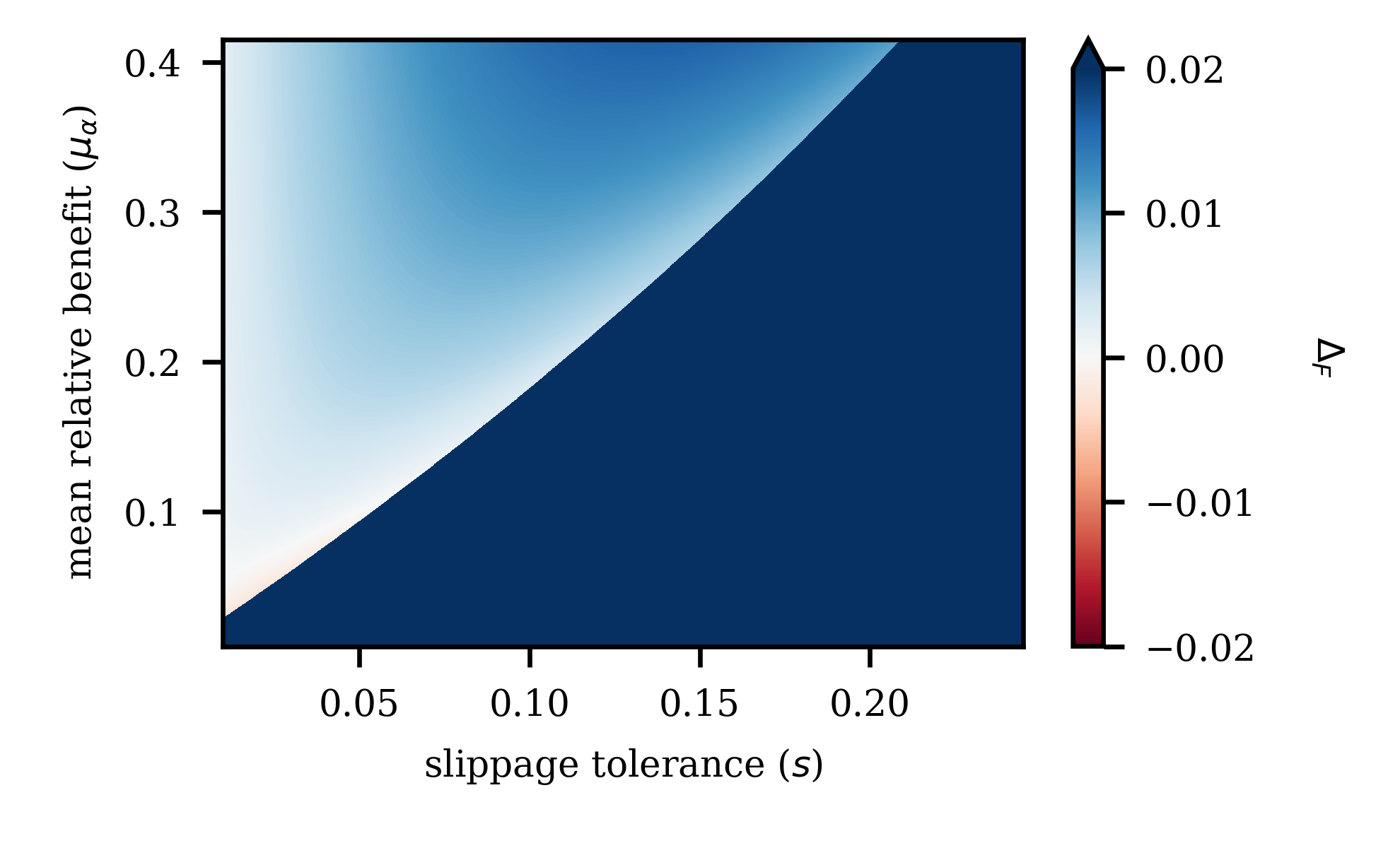}\vspace{-6pt}
    
    \caption{We set $k=3$.} \label{fig:signFeeDiffc3}
  \end{subfigure}\vspace{-4pt}
 \caption{Visualization  of $ \Delta_ F$ across both pools for a heterogeneous trader distribution $\psi_A ( \alpha)$ depending on mean relative benefit and the slippage tolerance. In blue areas the Nash equilibrium is $\text{Pool}_N$, in red areas, it is $\text{Pool}_W$, and in the white area in-between all liquidity distributions are Nash equilibria. $ \Delta_ F$ is cut off at 0.02 for better visibility and the dotted dark blue line visualizes where $\Delta_F = 0.01$. We set $x = 5,000,000$ $X$, $y = 5,000,000$ $Y$, $f= 0.003$ and $\omega=0.01$.} \label{fig:signFeeDiffC}
\end{figure}

\section{Omitted Strategy Proofs}\label{app:strategies}
\subsection{Sandwich Attack Profitability}
\profit*
\begin{proof}
    First, the sandwich attacker swaps $a^{\text{in}}_{x} $ and receives 
    $$a_y = - \int _{(1-p)x}^{(1-p)x+(1-f) a^{\text{in}}_{x}} \frac{-x\cdot y}{\xi^2} d \xi ,$$
    in the front-running transaction $A_F$. Then the trader sells $\delta _{x,W}$ and in return receives 
    $$ \tilde{\delta}_{y,W}  = - \int _{(1-p)x+(1-f) a^{\text{in}}_{x}}^{(1-p)x+(1-f) (a^{\text{in}}_{x}+ \delta_{x,W})} \frac{-x\cdot y}{\xi^2} d \xi .$$ 
    
    Finally, the sandwich attacker uses $a_y$ $Y$-tokens to buy $a^{\text{out}}_x$ $X$-tokens in its back-running transaction $A^B$. Due to the transaction fee $f$ being applied to the input, only $\tilde{a}_y = (1-f) a_y$ of the initially swapped $a_y$ re-enters the pool. Therefore, we write  $$ \tilde{a}_y= \int_{(1-p)x+(1-f) (a^{\text{in}}_{x}+  \delta_{x,W})}^{{(1-p)x+(1-f) (a^{\text{in}}_{x}+ \delta_{x,W})}- a^{\text{out}}_x} \frac{-x\cdot y}{\xi^2} d \xi,$$
    where the sign change in front of the integral is the result of $Y$-assets being returned to the pool. 
    
    The amount of $X$ the attacker holds after the transaction $a^{\text{out}}_x$  can be found by equating the two integrals for $a_y$ and $\tilde{a}_y$, using $\tilde{a}_y=(1-f)a_y$, and solving for $a^{\text{out}}_x$. This yields the profit of the sandwich attacker
    \begin{align*}
        &U^A\\
        &=a^{\text{out}}_x- a^{\text{in}}_x\\
        &=\tfrac{(1 - f)^2 a^{\text{in}}_{x}  ((1-p)x + (1 - f) (a^{\text{in}}_{x}  + \delta_{x,W}))^2}{
 ((1-p)x)^2 + (2 - f) (1 - f) (1-p) x \cdot  a^{\text{in}}_{x} + (1 - f)^3 a^{\text{in}}_{x}  (a^{\text{in}}_{x}  + \delta_{x,W})}- a^{\text{in}}_{x}.
    \end{align*}
\end{proof}

\smalldelta*
\begin{proof}
    The expression for $\delta_x^{\min}$ follows from Lemma~\ref{lem:profit} by solving $U^A=0$. The minimum transaction size for which a profitable sandwich attack exists is obtained by setting $a_x^\text{in}=0$.
\end{proof}

\front* 
\begin{proof} 
    We consider a sandwich attack with initial input $a^{\text{in}}_{x}$ to the front-running transaction. The output of the victim transaction selling $\delta_{x,W}$ becomes
    $$ \tilde{\delta}_{y,W} = - \int _{(1-p)x+(1-f) a^{\text{in}}_{x}}^{(1-p)x+(1-f) (a^{\text{in}}_{x}+ \delta_{x,W})} \frac{-x\cdot y}{\xi^2} d \xi .$$ 
  
    The victim's transaction will, however, only go through, if
    \begin{align*}
        &\tilde{\delta} _{y,W} \geq (1-s) \delta_{y,W} \\
        &- \int _{(1-p)x+(1-f) a^{\text{in}}_{x}}^{(1-p)x+(1-f) (a^{\text{in}}_{x}+ \delta_{x,W})} \frac{-x\cdot y}{\xi^2} d \xi \\ &\geq (1-s) \left(- \int _{(1-p)x}^{x+(1-f)\delta_{x,W}} \frac{-x\cdot y}{\xi^2} d \xi \right).
    \end{align*}
    Thus, the attacker's maximal input $a_x^s$ increases the slippage incurred by the victim to their tolerance, i.e., $\tilde{\delta} _{y,W} = (1-s) \delta_{y,W}$. Solving for $a_x^s$, we find that the maximal input is 
    \begin{align*}
        a_x^s  =& \frac{1}{2}\left( \frac{ \sqrt{\delta_{x,W}^2(1 - f)^2 +\frac{4(1-p)x((1-p)x+\delta_{x,W}(1-f))}{1-s}}}{1-f}\right.\\&\hspace{0.7cm}\left. -\frac{-2(1-p)x }{1-f}-\delta_{x,W}\right).
    \end{align*}
\end{proof}

\subsection{Trade Sizes}

\tradesizesoph* 
\begin{proof}
    Without loss of generality, we maximize the trader's utility in each pool independently and start with $\text{Pool}_N$. The trader's utility in $\text{Pool}_N$ is given by 
    $$U^{T}_N =(1+\alpha)\frac{(1-f)\delta_{x,N}p\cdot y}{(1-f)\delta_{x,N}+p\cdot x }- \frac{y}{x}\delta_{x,N} .$$
    We differentiate the trader's utility in $\text{Pool}_N$, $U^{T}_N$, with respect to the transaction size $\delta_{x,N}$ to find the transaction size $\delta_{x,N} ^{\text{opt}}$ maximizing the trader's utility. We obtain
    \begin{align*}
     \partial _{\delta_{x,N}} U^{T}_N=& \frac{(1+\alpha)(1-f) p^2 \cdot x \cdot y}{(\delta_{x,N} (1-f)+p\cdot x)^2}-\frac{y}{x},
    \end{align*}
    and the two zero crossing of $\partial _{\delta_{x,N}} U^{T}_N$ are:
    \begin{align*}
        p\cdot x(\pm \sqrt{(1+\alpha)(1-f)}-1) /(1-f). 
    \end{align*}
    For our parameters, $x,y,\alpha>0$, $0<f<1$, $0\leq p\leq 1$, the second derivative, $\partial^2 _{\delta_{x,N}} U^{T}_N$, is only negative for the following zero crossing
    \begin{align*}
        \delta_{x,N}^{\max}=p\cdot x (\sqrt{(1+\alpha)(1-f)} -1)/(1-f). 
    \end{align*}
    Thereby, $\delta_{x,N}^{\max}$ maximizes the traders utility. The trader optimally sells $\delta_{x,N} ^{\text{opt}}=\max(0,\delta_{x,N}^{\max})$ in $\text{Pool}_N$. 
    
    We proceed analogously as above for $\text{Pool}_W$ and find the trader the optimally places $\delta_{x,W} ^{\text{opt}}=\max(0,\delta_{x,W}^{\max})$ in $\text{Pool}_W$. In the previous, 
    $$\delta_{x,W}^{\max} =(1-p) x( \sqrt{(1+\alpha)(1-s)(1-f)}-1)/(1-f).$$
    The trade inputs to maximize the trader's utility can, thus, be determined analytically and are given by $\delta_{x,N} ^{\text{opt}}=\max(0,\delta_{x,N}^{\max})$ and $\delta_{x,W} ^{\text{opt}}=\max(0,\delta_{x,W}^{\max})$. 
\end{proof}
\section{Omitted Game Equilibria Proofs}\label{app:game}
\subsection{Homogeneous Traders}

\strongNash*
\begin{proof}
    If the fees gradient is non-zero ($\partial _p F \neq 0$), the maxima is located at either corner point of the interval, as fees are proportional to $p$ (cf. Lemma~\ref{lem:totalfee}),  and the fees gradient is non-zero, the maxima are located at either corner point of the interval.     
    Otherwise, if the fees gradient $\partial _p F$ is zero, the fees across the entire interval are constant. Therefore, all liquidity distributions are $\varepsilon$-equilibria.
\end{proof}

\softNash*
\begin{proof}
    For a liquidity distribution to qualify as an $\varepsilon$-equilibrium, no liquidity provider must see a possibility to increase their expected fees by more than than a factor $1+\varepsilon$ through adjusting their liquidity distribution. Further, we know from Lemma~\ref{lem:strongNash} that a liquidity provider receives the most fees either when all their liquidity is in $\text{Pool}_N$ or all their liquidity is in $\text{Pool}_W$. Thus, the maximum relative increase to an LP's fees, with current liquidity distribution $(p_i l_i L,(1-p_i) l_i L)$, is given as
    $$\frac{\max\{F (f,\alpha,s,y,0,\omega) ,F (f,\alpha,s,y,1,\omega)  \} }{  F \left(f,\alpha,s,y,p_i,\omega\right)}-1.$$
    In case the previous fraction does not exceed $\varepsilon$ for any liquidity provider, the current distribution is a Nash equilibrium. 
\end{proof}

\subsection{Heterogeneous Traders}
\epsNashdist*
\begin{proof}
    We proceed similarly to Theorem~\ref{thm:softNash} and note that as long as no liquidity provider must see a possibility to increase their expected fees by more than than a factor $1+\varepsilon$ through adjusting their liquidity distribution, the configuration is a Nash equilibrium. The fees received by a liquidity provider $LP_i$ are given by 
    $$\sum _\alpha \psi_A(\alpha)F \left(f,\alpha,s,y,p_i,\omega\right),$$
    where $F \left(f,\alpha,s,y,p_i\right)$ is given by Lemma~\ref{lem:totalfee}. We, thus, follow that the fees received by liquidity provider $LP_i$ are also proportional to $p_i$ in the heterogeneous case. Further, it holds that the liquidity provider receives the most fees either when all their liquidity is in $\text{Pool}_N$, all their liquidity is in $\text{Pool}_W$ or the fees are constant for all liquidity distribution. Thus, the maximum relative increase to an LP's fees, with current liquidity distribution $(p_i l_i L,(1-p_i) l_i L)$, is given as
    \begin{gather*}
        \frac {\max\left\{\sum _\alpha \psi_A(\alpha)F (f,\alpha,s,y,0,\omega)  ,\sum _\alpha\psi_A(\alpha)F (f,\alpha,s,y,1,\omega)  \right\}}{ \sum _\alpha \psi_A(\alpha)F \left(f,\alpha,s,y,p_i,\omega\right)  } \\- 1<  \varepsilon.
    \end{gather*}    
    
    In case the previous fraction does not exceed $\varepsilon$ for any liquidity provider, the current distribution is a Nash equilibrium. 
\end{proof}
%%%%%%%%%%%%%%%%%%%%%%%%%%%%%%%%%%%%%%%%%%%%%%%%%%%%%%%%%%%%%%%%%%%%%%%%%%%%%%%%
\end{document}